\documentclass[useAMS,referee,usenatbib]{biom}
\usepackage{graphicx}
\usepackage{graphics}
\usepackage{amsmath}
\usepackage{amssymb}

\def\bSig\mathbf{\Sigma}

\newcommand\independent{\protect\mathpalette{\protect\independenT}{\perp}}
\def\independenT#1#2{\mathrel{\rlap{$#1#2$}\mkern2mu{#1#2}}}

\newcommand{\inp}{\stackrel{p}{\longrightarrow}}
\newcommand{\inD}{\stackrel{D}{\longrightarrow}}

\newcommand{\bW}{\mathcal{W}}
\newcommand{\bd}{\mathbf{d}}
\newcommand{\bX}{\boldsymbol{X}}
\newcommand{\bA}{\boldsymbol{A}}
\newcommand{\bE}{\boldsymbol{E}}
\newcommand{\bM}{\boldsymbol{M}}
\newcommand{\ba}{\boldsymbol{a}}
\newcommand{\bH}{\boldsymbol{H}}
\newcommand{\bx}{\boldsymbol{x}}

\newcommand{\bh}{\boldsymbol{h}}
\newcommand{\bI}{\boldsymbol{I}}
\newcommand{\Xbar}{\overline{\bX}}
\newcommand{\xbar}{\overline{\bx}}
\newcommand{\abar}{\overline{\ba}}
\newcommand{\Abar}{\overline{\bA}}
\newcommand{\bbeta}{\boldsymbol{\beta}}
\newcommand{\hatbeta}{\widehat{\bbeta}}
\newcommand{\hatmu}{\widehat{\mu}}
\newcommand{\btheta}{\boldsymbol{\theta}}

\newcommand{\hattheta}{\widehat{\btheta}}
\newcommand{\hatnu}{\widehat{\nu}}
\newcommand{\hatsigma}{\widehat{\sigma}}
\newcommand{\tiltheta}{\widetilde{\btheta}}
\newcommand{\tilbeta}{\widetilde{\bbeta}}
\newcommand{\bphi}{\boldsymbol{\phi}}

\newcommand{\bgamma}{\boldsymbol{\gamma}}

\newcommand{\hatSig}{\widehat{\boldsymbol{\Sigma}}}

\newcommand{\tildeV}{\widetilde{V}}

\newcommand{\dbar}{\overline{d}}

\newcommand{\calV}{\mathcal{V}}

\newcommand{\calA}{\mathcal{A}}
\newcommand{\calP}{\mathcal{P}}

\newcommand{\hatP}{\widehat{\calP}}
\newcommand{\calH}{\mathcal{H}}

\newcommand{\calN}{\mathcal{N}}

\newcommand{\Data}{\mathfrak{D}}

\newcommand{\sumiN}{\sum^N_{i=1}}

\newcommand{\hatQ}{\widehat{Q}}

\newtheorem{thm}{Theorem}

\title[Thompson Sampling for SMARTs]{Adaptive Randomization Methods
  for Sequential Multiple Assignment Randomized Trials (SMARTs) via
  Thompson Sampling}

\author{Peter Norwood$^{1,*}$\email{p.norwood@quantumleaphealth.org}, 
Marie Davidian$^{2**}$\email{davidian@ncsu.edu}, and 
Eric Laber$^{3***}$\email{eric.laber@duke.edu} \\
$^{1}$Quantum Leap Healthcare Collaborative, San Francisco, California, U.S.A.\\
$^{2}$Department of Statistics, 
North Carolina State University, Raleigh, North Carolina, U.S.A. \\
$^{3}$Department of Statistical Science, 
Duke University, Durham, North Carolina, U.S.A.}

\begin{document}






\begin{abstract}
  Response-adaptive randomization (RAR) has been studied extensively
  in conventional, single-stage clinical trials, where it has been
  shown to yield ethical and statistical benefits, especially in
  trials with many treatment arms.  However, RAR and its potential
  benefits are understudied in sequential multiple assignment
  randomized trials (SMARTs), which are the gold-standard trial design
  for evaluation of multi-stage treatment regimes.  We propose a suite
  of RAR algorithms for SMARTs based on Thompson Sampling (TS), a
  widely used RAR method in single-stage trials in which treatment
  randomization probabilities are aligned with the estimated
  probability that the treatment is optimal.  We focus on two common
  objectives in SMARTs: (i) comparison of the regimes embedded in the
  trial, and (ii) estimation of an optimal embedded regime.  We
  develop valid post-study inferential procedures for treatment
  regimes under the proposed algorithms.  This is nontrivial, as (even
  in single-stage settings) RAR can lead to nonnormal limiting
  distributions of estimators. Our algorithms are the first for RAR in
  multi-stage trials that account for nonregularity in the estimand.
  Empirical studies based on real-world SMARTs show that TS can
  improve in-trial subject outcomes without sacrificing efficiency for
  post-trial comparisons.  \vspace*{0.3in}
\end{abstract}

%



\maketitle

\section{Introduction}
\label{s:intro}


A treatment regime, also known as an adaptive treatment strategy, is a
sequence of decision rules, one for each key decision point in a
patient's disease progression, that maps accumulated information on a
patient to a recommended treatment. \citep[][]{tsiatis2020dynamic}.
Sequential multiple assignment randomized trials (SMARTs) are the gold
standard for the study of treatment regimes \citep{lavori_dawson,murphy_2005} and have been applied successfully
in a range of areas, including cancer, addiction, HIV/prevention, and
education \citep[e.g.,][]{kidwell_2014,
bigirumurame2022sequential,lorenzoni2023use}.  A SMART involves
multiple stages of randomization, each stage corresponding to a
decision point, where the sets of possible treatments may depend on
baseline and interim information.



Figure~\ref{fig:example_smart} depicts a two-stage SMART
to evaluate behavioral intervention strategies for cancer pain
management \citep{tammy_smart_2023}.
Subjects were randomized at baseline with equal probability to one of
two first-stage interventions: Pain Coping Skills Training (PCST) with
five sessions (PCST-Full), or one session (PCST-Brief).  At the end of
the first stage, subjects were classified as responders if they
experienced a 30\% reduction in their pain score from baseline and as
nonresponders otherwise.  Subjects were then assigned with equal
probability to one of two second-stage interventions depending on
their first-stage intervention and response status.  This design, like any fixed
(nonadaptive) SMART design, can be represented as a single-stage trial
in which subjects are randomized at baseline among a set of fixed
regimes known as the SMART's embedded regimes \citep[][Chapter
9]{tsiatis2020dynamic}.  The cancer pain SMART has eight embedded
regimes determined by the stage 1 treatment, stage 2 treatment for
responders, and stage 2 treatment for nonresponders, e.g., one
embedded regime assigns PCST-Full initially followed by no further
treatment if the subject responds and another round of PCST-Full
otherwise.  A key goal was to evaluate the embedded regimes on the
basis of mean percent reduction in pain from baseline at the end of
stage 2 and to identify a regime that yields the greatest
reduction.
\begin{figure}[t]
\centering
\includegraphics[width=12cm]{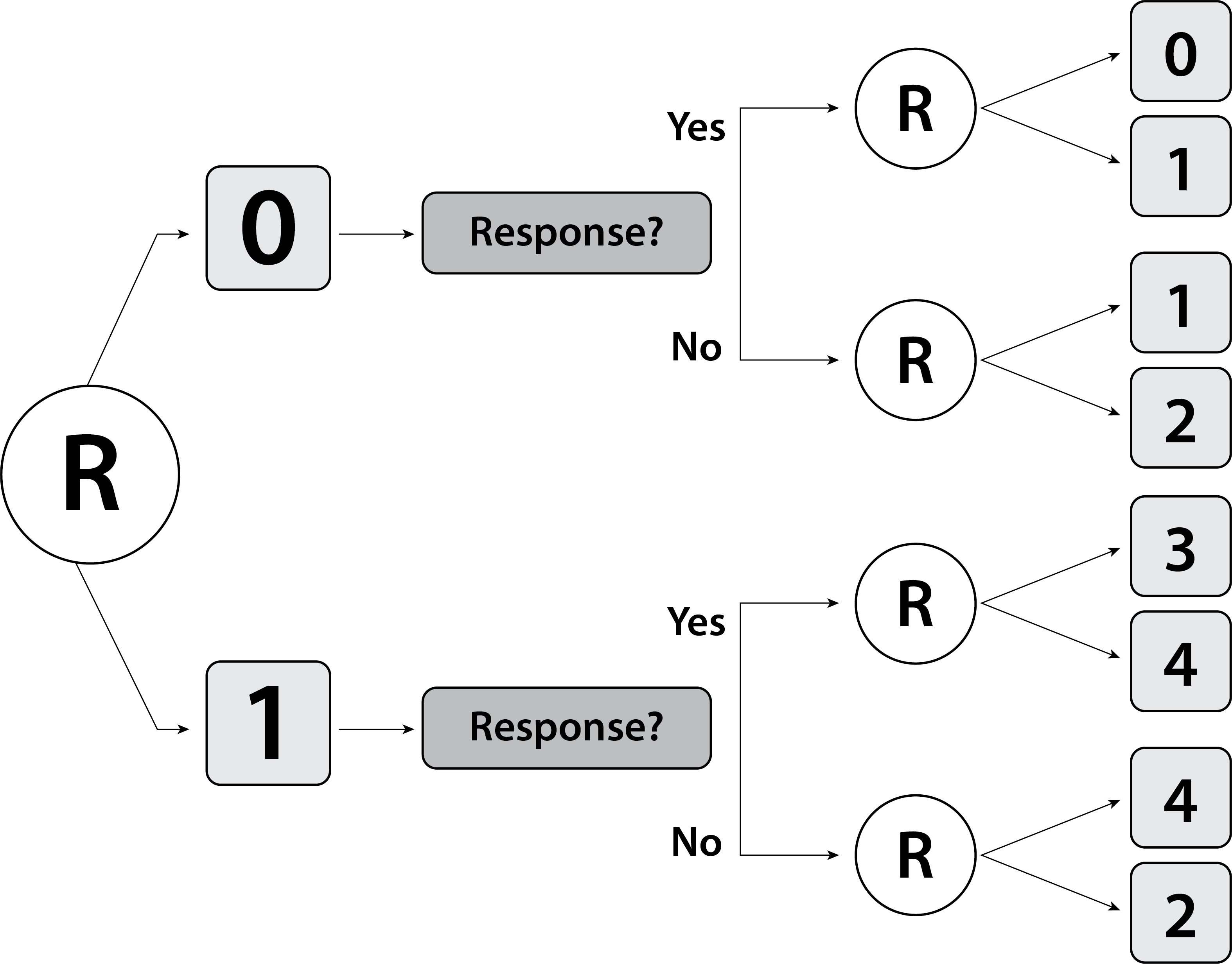}
\caption{\label{fig:example_smart} Schematic depicting the design of
  the SMART for evaluation of behavioral interventions for cancer pain
  management.  The eight embedded regimes implied by the design are of
  the form ``Give $a$ initially; if response, give $b$, otherwise if
  nonresponse give $c$,'' where regimes 1,\ldots, 8 correspond to
  $(a,b,c) =$ (0,0,1), (0,0,2), (0,1,2), (0,1,1), (1,3,4), (1,3,2),
  (1,4,2), (1,4,4), respectively.}
\end{figure}

As in the preceding example, most SMARTs use fixed randomization
probabilities at each stage.  Although fixed and balanced
randomization in SMARTs yields high power for comparing treatments and
treatment regimes \citep[][]{murphy_2005}, ethical and pragmatic
considerations suggest that updating randomization probabilities based
on accruing information can improve outcomes for trial subjects,
increase enrollment, and decrease dropout \citep{FDA2019}.
Response-adaptive randomization (RAR) uses accumulating information to
skew randomization probabilities toward promising treatments and has
long been used in single-stage randomized clinical trials 
\citep[RCTs,][]{kim2011battle,berry2015brave,ISPY2}.
There is a lengthy literature on RAR methods; reviews for
single-stage RCTs include \citet{hu2006theory},
\citet{berry2010bayesian}, and \citet{atkinson2013randomised}.  RAR
procedures are typically formalized as a multi-arm bandit in which
each treatment option is a bandit ``arm''
\citep[][]{berry1985bandit,villar_bowden_wason_2015,lattimore2020bandit};
strong theoretical and empirical support for bandits has grown from
their use in applications such as mHealth and
telehealth \citep[][]{tewari2017ads,liu2023microrandomized}.
In single-stage RCTs, the advantages of RAR are most pronounced in
trials with a large number of treatments because RAR algorithms
oversample more favorable treatments and undersample
less favorable ones \citep{berry_nature}.  RAR also presents
challenges, e.g., early evidence can lead the randomization to become
``stuck'' favoring suboptimal treatments \citep*{ThallFoxChapter}, and
there have been spirited debates about the potential pitfalls of RAR
\citep[][]{viele2020comparison,proschan2020resist,villar2021temptation}.

As the number of treatments and regimes evaluated in a SMART can be
large, integration of RAR into SMARTs could have significant benefits,
yet development is limited.  \citet*{cheung_chakraborty_davidson_2014}
propose SMART-AR, an RAR method based on Q-learning
\citep[e.g.,][Section~5.7.1]{tsiatis2020dynamic} that adapts
randomization probabilities so that treatments with large estimated
Q-functions are more likely to be selected.  However, this approach
does not account for uncertainty in the estimated Q-functions nor the
potential information gain associated with each treatment, so the
resulting design may be inefficient.  \citet*{wang_wu_wahed_2021}
propose RA-SMART, an RAR method for two-stage SMARTs 
with the same treatments at each stage.  The method does not
incorporate delayed effects, i.e., does not account for treatment
effects at stage 2 that depend on stage 1 treatment.  Because the
possibility of interactions among treatments at different stages is
key to the development of treatment regimes, RAR schemes that
acknowledge delayed effects are desirable.

In this article, we propose RAR approaches for SMARTs based on
Thompson sampling \citep[TS,][]{thompson1933likelihood}, a popular bandit algorithm also known as
probability matching
\citep[][]{agrawal2012analysis,thompson_tutorial}. The idea underlying
TS is that a treatment's randomization probability should be based on
confidence that the treatment is optimal.  Classical TS measures this
confidence with a posterior probability; aligned with standard SMART
methodology, we adopt a frequentist perspective in which confidence is
assessed using a confidence distribution \citep[][]{xie2013confidence}.
The proposed methods are applicable to SMARTs 
used to evaluate a fixed set of treatment regimes and/or
identify an optimal regime and, by construction, account for
delayed effects.

In Section~\ref{s:background}, we present the statistical framework
and a variant of TS for SMARTs, and we use the latter in
Section~\ref{s:methods} to construct the proposed RAR approaches.  One
class of methods randomizes subjects ``up-front'' to entire
embedded regimes, while the other sequentially randomizes subjects at each decision
point.  We propose estimators for the marginal mean outcome
under a regime in Section~\ref{s:methods} and argue in
Section~\ref{s:theory} that they are consistent and asymptotically
normal; this is nontrivial, as adaptive randomization can lead to
nonnormal limits for plug-in estimators
\citep{Hadad2021,ZhangMest2021}.  
Simulations demonstrating performance are reported in Section~\ref{s:sims}.

  
\section{Background and preliminaries}
\label{s:background}

\vspace*{-0.1in}

\subsection{Notation and assumptions}
\label{ss:SMART}

\vspace*{-0.05in}

We first review SMARTs with fixed (i.e., nonadaptive) randomization.
Consider a SMART with $K$ stages, where, at stage $k=1,\ldots,K$,
$\mathcal{A}_k$ is the (finite) set of available treatment options.
Let $A_k \in \mathcal{A}_k$ be the treatment assigned at stage $k$,
and let $\overline{A}_k = (A_1,\ldots, A_k) \in \calA_1 \times \cdots \times \calA_k$
denote all treatments
given through stage $k$. 
Let $\bX_1$ be a set of baseline subject variables collected prior to
administration of stage 1 treatment, and let
$\bX_k$ 
comprise variables collected between stages $k-1$ and $k$,
$k=2,\ldots,K$. Define $\overline{\bX}_k = (\bX_1,\dots,\bX_k)$,
$k=1,\ldots,K$; and  let $\bH_1=\bX_1$ and
$\bH_k = (\bX_1,A_1,\bX_2,A_2,\dots,A_{k-1},\bX_k) =
(\Xbar_k,\Abar_{k-1})$, $k=2,\ldots,K$, denote the information at the
time $A_k$ is assigned, with $\mathcal{H}_k$ denoting the domain of
$\bH_k$, $k=1,\ldots, K$.  Let
$Y$ be the real-valued outcome of interest, which is measured or
constructed from $\overline{\bX}_{K}$ and information measured after
stage $K$ at some specified follow-up time,
coded so that larger values are more favorable.  Elements of $\bX_k$
may include current and past measures of the patient's health status,
previous measures of the outcome, and response status.  For a subject
with history $\bH_k=\bh_k$ at stage $k$,
$\Psi_k(\bh_k) \subseteq \calA_k$ is the set of feasible treatment
options in $\calA_k$, where $\Psi_k$ maps $\calH_k$ to subsets of
$\calA_k$ \citep[Section 6.2.2]{tsiatis2020dynamic}. E.g., at
stage $k=2$ of the cancer pain SMART, $\Psi_2(\bh_2) = \{1, 2\}$ for a
subject who does not respond to treatment 0 given at the first stage.  At
stage $k$, a subject is randomized to option $a_k \in \Psi_k(\bh_k)$
with probability $P(A_k = a_k | \bH_k = \bh_k)$; often,
$P(A_k=a_k|\bH_k=\bh_k) = 1/|\Psi_k(\bh_k)|$ for all $k$.

A decision rule $d_k: \mathcal{H}_k \rightarrow \mathcal{A}_k$ maps an
individual's history to a recommended treatment at stage $k$, where
$d_k(\bh_k) \in \Psi_k(\bh_k)$ for all $\bh_k \in \calH_k$.  A
treatment regime $\bd = \left( d_1,\dots,d_K \right)$ is a sequence
of decision rules, where $d_k$ is the rule for stage $k$,
$k=1,\ldots,K$.  The mean outcome that would be achieved if the
population were to receive treatment 
according to $\bd$ can be characterized in terms of
potential outcomes. 
For any
$\overline{\ba}_{k-1} = (a_1,\ldots, a_{k-1}) \in \mathcal{A}_{k-1}$, let
$\bX_k^*(\abar_{k-1})$ denote the potential subject information that
would accrue between stages $k-1$ and $k$, and define
$\Xbar_k^*(\abar_{k-1}) =
\{\bX_1,\bX_2^*(a_1),\ldots,\bX_k^*(\abar_{k-1})\}$, $k=2,\ldots,K$.
The potential history at stage $k\ge 2$ is thus
$\bH_k^*(\abar_{k-1}) = \{\Xbar_k^*(\abar_{k-1}), \overline{\ba}_{k-1} \}$, with 
$\bX_1^*(\overline{\ba}_0) = \bH_1^*(\overline{\ba}_0) = \bH_1$.  For any
$\abar_K\in\overline{\mathcal{A}}_K$, let $Y^*(\abar_K)$ be the
potential outcome that would be achieved under $\abar_K$.
For $k=2,\ldots, K$, the potential accrued information between stages $k-1$ and $k$ under regime $\bd$ is then
$\bX^*_k(\bd) = \sum_{\abar_{k-1} \in \overline{\calA}_{k-1}} \bX^*_k(\abar_{k-1})
\prod^{k-1}_{v=1} I [ d_v\{ \bH_v^*(\abar_{v-1}) \} = a_v]$, and
the potential outcome under $\bd$ is 
$Y^*(\bd) = \sum_{\abar_K \in \overline{\calA}_K} Y^*(\abar_K)
  \prod^K_{k=1} I[ d_k\{ \bH_k^*(\abar_{k-1}) \} = a_k]$,
where $I(\cdot)$ is the indicator function.  The
mean outcome for regime $\bd$, known as the value of $\bd$, is 
$\calV(\bd) = E\{ Y^*(\bd)\}$.

Define
$\bW^* = \{\bX_2^*(a_1),\ldots,\bX_K^*(\abar_{K-1}),Y^*(\abar_K)$ for all
$\abar_K \in \calA_K\}$ to be the set of all potential outcomes. 
For a given regime $\bd$, identification of $\calV(\bd)$ is possible
from the data $(\Xbar_K,\Abar_K,Y)$ under the following assumptions,
which are discussed extensively elsewhere \citep[e.g., ][Section 6.2.4]{tsiatis2020dynamic} and which we adopt: (i) consistency,
$Y = Y^{*}(\Abar_K)$, $\bH_k = \bH_k^*(\Abar_{k-1})$, $k=1,\ldots,K$;
(ii) positivity, $P(A_k = a_k | \bH_k = \bh_k) >0$ for all
$a_k \in \Psi_k(\bh_k)$ and all $\bh_k \in \calH_k$, $k=1,\ldots,K$;
and (iii) sequential ignorability,
$\bW^* \independent A_k \, | \, \bH_k$, $k=1,\ldots,K$, where
``$\independent$'' denotes statistical independence.  We also assume
that there is no interference among individuals nor multiple versions
of a treatment.  In a SMART with nonadaptive randomization, (ii) is
true by design, and (iii) is guaranteed by randomization; with RAR,
(ii) holds if constraints are imposed on the randomization
probabilities (discussed shortly) and (iii) holds because
randomization probabilities are known features of the history.

Denote the $m$ regimes embedded in a SMART
as $\bd^1,\ldots,\bd^m$.  A common primary analysis is the comparison of
$\calV(\bd^j)$, $j=1,\ldots,m$, i.e., the mean outcomes that would be
achieved if the population were to receive treatments according to
$\bd^1,\ldots,\bd^m$.  Another common analysis is identification of an
optimal embedded regime, $\bd^{\mathrm{opt}}$,
satisfying $\calV(\bd^{\mathrm{opt}}) \geq \calV(\bd^j)$, $j=1,\ldots,m$.


\subsection{Subject accrual and progression processes}
\label{ss:accrual}

\vspace*{-0.0in}

Ordinarily, it is assumed that subjects enroll in a clinical trial by
a completely random process over a planned accrual period.  Under such
a process, subjects in a SMART progress through the stages
in a staggered fashion.  Thus, at any point during the trial, enrolled
subjects will have reached different stages, with only some having the
outcome $Y$ ascertained.  By reaching stage $k$, we mean that a
subject has completed the previous stages and that $A_k$ has been
assigned. E.g., in a SMART with $K=2$ stages, at any time, there may
be subjects who have received only a first-stage treatment, for whom
$(\bX_1,A_1)$ is available; subjects who have reached stage 2 but have
not completed follow up, for whom $(\bX_1,A_1,\bX_2,A_2)$ is
available; and subjects who have completed the trial, for whom
$(\bX_1,A_1,\bX_2,A_2,Y)$ is available.  


Ideally, randomization probabilities would be updated each time a new
randomization is needed.  However, for logistical or other reasons, it
may be more feasible in practice to update the probabilities on a
pre-set schedule.  For definiteness, we consider a SMART that will
enroll $N$ subjects over $T$ weeks, with randomization probabilities
updated weekly.  We develop two randomization schemes: the first,
which we term up-front, randomizes subjects to an embedded regime at
baseline; the second, which we term sequential, randomizes subjects at
each stage.  While these two schemes are equivalent in a SMART with fixed randomization probabilities, the sequential scheme under RAR allows a subject's randomization probabilities to depend on 
information collected as they progress through the
trial. 

We assume that a group of subjects of random size enrolls at each week
$t$ and, depending on the randomization scheme, up-front or
sequential, requires assignment to either an entire regime (up-front)
or a stage 1 treatment (sequential) according to probabilities that
are updated at $t$ based on the accrued data from subjects previously
enrolled at weeks $t-1, t-2,\ldots,0$.  Each of these subjects is
assigned a regime (up-front) or $A_1$ (sequential) using these
probabilities.  At each subsequent stage $k=2,\ldots,K$, we assume
that there is a delay, measured in weeks, between when $A_{k-1}$ is
assigned and when $\bX_k$ is ascertained and $A_k$ is assigned, as
well as a follow-up period after $A_K$ is assigned but before $Y$ is
recorded.  Thus, under sequential randomization, these subjects will
require randomization to $A_k$, $k = 2,\ldots,K$, at a future week
$(t+s)$, say, using probabilities based on the accrued data from
subjects previously enrolled at or before week $(t+s-1)$.

To represent the data available on subjects who are already enrolled
at any week $t$, let $\Gamma_{t} \in \{0,1\}$ be the indicator that a
subject has enrolled in the SMART and been assigned to a regime or  stage 1
treatment by week $t$ (i.e., at a week $\leq t$). For such
subjects, let $\tau \leq t$ be the week of enrollment and assignment;
$\kappa_t \in \{1,\ldots,K\}$ be the most recent stage reached by week
$t$; and $\Delta_{t} \in \{0, 1\}$ be the indicator that a subject has
completed follow up by week $t$, so that $Y$ has been observed.  The
data available on a subject at week $t$ are then
$D_{t} = \Gamma_{t}\{ 1, \tau,\kappa_{t}, \bX_1, A_1,
I(\kappa_{t}>1)\bX_2, I(\kappa_{t}>1)A_2, \ldots,
I(\kappa_{t}>K-1)\bX_K, I(\kappa_{t}>K-1)A_K, \Delta_{t}, \Delta_{t}Y
\}$; $D_t$ is null if a subject has not yet enrolled by $t$.  The
accrued data from all previously-enrolled subjects that can inform the
randomization probabilities at week $t$ are then
$\Data_{t-1} = \{ D_{t-1,i}, \forall i \mbox{ such that } \Gamma_{t-1,i}
= 1\}$.  


\subsection{Thompson sampling}
\label{ss:thompson}

\vspace*{-0.05in}

The central idea of TS is to map a treatment's randomization
probability to the ``belief'' that it is optimal among the available
options.  This belief is represented conventionally by the posterior
probability that a treatment is optimal in the sense that it optimizes
expected outcome \citep{thompson1933likelihood}.  For example, in the
up-front case (and similarly for a single-stage treatment), let
$\widehat{\rho}_{t}^j \in [0,1]$ be the estimated belief that
treatment regime $j=1,\ldots,m$ is optimal at week $t$ based on the
accrued data $\Data_{t-1}$ from previously-enrolled subjects.  Let
$r_{t}^j$ be regime $j$'s randomization probability for subjects
needing randomization at week $t$.  Ordinarily, $r_{t}^j$ is taken to
be a monotone function of $\widehat{\rho}_{t}^j$; a popular choice is
$r_{t}^j= (\widehat{\rho}_{t}^j)^{c_t} / \sum_{v=1}^m
(\widehat{\rho}_{t}^v)^{c_t},$ where $c_t \in [0,1]$ is a damping
constant \citep{practical_thompson}.  Smaller values of $c_t$ pull the
probabilities toward uniform randomization; higher values pull them
closer to the beliefs.  The $c_t$ can be the same for all $t$ or 
increase with $t$ to impose greater adaptation as data accumulate.
Because aggressive adaptation can lead to randomization probabilities
approaching 0 or 1 for some regimes or treatments, limiting
exploration of all options,
one may impose clipping constants, i.e., lower/upper bounds on the
probabilities \citep{zhang2021inference}.  This practice is 
consistent with the positivity assumption.

In a fully Bayesian formulation, beliefs $\widehat{\rho}_{t}^j$ are
estimated posterior probabilities based on $\Data_{t-1}$.  We propose a
frequentist analog based on the so-called confidence distribution
\citep[][]{xie2013confidence}. Here, we provide a basic overview
of the approach; details for the up-front and sequential algorithms
are in subsequent sections.  Let $\btheta$ be the vector of parameters
that, with a subject's current history, determines a subject's optimal
treatment; e.g., under up-front randomization, 
$\btheta = (\theta^1,\ldots,\theta^m)^T =
\{\calV(\bd^1),\ldots,\calV(\bd^m)\}^T$.  Let $\hatP_{\btheta,t}$ be
the estimated confidence distribution for $\btheta$ at week $t$ based
on $\Data_{t-1}$, e.g., the estimated asymptotic distribution of some
estimator $\hattheta_t$ for $\btheta$, and let
$\tiltheta^b_t = (\widetilde{\theta}^{1,b}_{t},\ldots,
\widetilde{\theta}^{m,b}_{t})^T$, $b=1,\ldots,B$, be independent draws
from $\hatP_{\btheta,t}$.  Then for $j=1,\ldots,m$,
$\widehat{\rho}_{t}^j = B^{-1} \sum_{b=1}^B I \{
\widetilde{\theta}^{j,b}_{t}=\max(\widetilde{\theta}^{1,b}_{t},\ldots,
\widetilde{\theta}^{m,b}_{t}) \}$.

At the onset of a SMART, a burn-in period of nonadaptive
randomization may be required from which to obtain an initial estimate
of $\btheta$ and associated confidence distribution to be used for the
first update of randomization probabilities.  The burn-in may be
characterized in terms of numbers of subjects who have completed each
stage or the calendar time elapsed since the start of the trial;
examples are presented in Section~\ref{s:sims}.  In what follows,
$t^*$ is the calendar time at which burn-in is complete, so that
adaptation starts at week $t = t^*+1$.

\vspace*{-0.25in}

\section{Adaptive randomization for SMARTs using Thompson sampling}
\label{s:methods}


We now present the proposed up-front and sequential TS approaches to
RAR for SMARTs.  As above, under up-front randomization, subjects are
randomized once, at enrollment, to an embedded regime,
which they follow through all $K$ stages.  Thus, the ``path'' that a
subject will take through the trial is determined by $\Data_{t-1}$.
Up-front randomization is logistically simpler but does not use additional data that have accumulated as the subject progresses through the trial. Up-front randomization is preferred when simplicity of implementation
is a priority or when enrollment is expected to be slow relative to
the time to progress through the trial, so that the amount of new data
accumulating before a subject completes all $K$ stages is modest.
Under sequential randomization, subjects are randomized at each stage,
so that up-to-date information from previous subjects informs
randomization probabilities as a subject progresses, but involves
greater logistical complexity.


\subsection{Up-front randomization among regimes}
\label{ss:embedded}

\vspace*{-0.05in}

To randomize newly-enrolled subjects at week $t$ to the embedded
regimes $\bd^j$, $j=1,\ldots,m$, we require an estimator $\hattheta_t$
for
$\btheta
=(\theta^1,\ldots,\theta^m)^T=\{\calV(\bd^1),\ldots,\calV(\bd^m)\}^T$
based on $\Data_{t-1}$ with which to construct a confidence
distribution and thus randomization probabilities $r_{t}^j$.  Aligned with standard methods for SMARTs, we focus on
inverse probability weighted (IPW) and augmented IPW (AIPW) estimators
for $\theta^1,\ldots,\theta^m$
\citep[Sections~6.4.3-6.4.4]{tsiatis2020dynamic}.  As we discuss in
Section~\ref{s:theory}, under any form of RAR, the asymptotic
distribution of these estimators based on the data at the end of the
trial need not be normal; thus, we adapt the approach of
\citet{ZhangMest2021} and propose weighted versions of these
estimators, where the weights are chosen so that the weighted
estimators are asymptotically normal.  For stratified sampling, the
formulation applies within each stratum.

For subjects who enroll at week $t$, define
$\eta_{t,1}(a_1, \bx_1, \Data_{t-1}) = P(A_1 = a_1| \bX_1 = \bx_1,
\Data_{t-1}) = P(A_1 = a_1| \bH_1 = \bh_1, \Data_{t-1})$ and, for $k = 2,\ldots, K$,
$\eta_{t,k}(a_k, \xbar_k,\abar_{k-1},\Data_{t-1}) = P(A_k = a_k| \Xbar_k= \xbar_k,\Abar_{k-1}=\abar_{k-1},
\Data_{t-1}) =P(A_k = a_k| \bH_k=\bh_k,
\Data_{t-1})$.  For regime $\bd$, let $\dbar_1(\bx_1) = d_1(\bx_1)$,
$\ldots$, $\dbar_k(\xbar_k) = [d_1(\bx_1),d_2\{\xbar_2,d_1(\bx_1)\}$,
$\ldots$, $d_k\{\xbar_k, \dbar_{k-1}(\xbar_{k-1})\}]$, $k=3,\ldots,K$.
For each $\bd^j$, let $C^j = I\{\Abar_K = \dbar_K^j(\Xbar_K)\}$ be
the indicator that a subject's experience through all $K$ stages is
consistent with receiving treatment using $\bd^j$.  Each enrolled subject in
$\Data_{t-1}$ was randomized at week $\tau \leq t-1$ based on
$\Data_{\tau-1} \subseteq \Data_{t-1}$, thus using $r_{\tau}^j$,
$j=1,\ldots,m$.  For each, define
$\pi_{\tau, 1}^j(\bx_1) = \eta_{\tau,1}\{d_1^j(\bx_1),
\bx_1,\Data_{\tau-1}\}$,
$\pi_{\tau,k}^j(\xbar_k) =
\eta_{\tau,k}[d_k^j\{\xbar_k,\dbar_{k-1}^j(\xbar_{k-1})\},
\xbar_k,\dbar_{k-1}^j(\xbar_{k-1}),\Data_{\tau-1}]$, $k=2,\ldots,K$.
The weighted IPW (WIPW) estimator for $\theta^j$ is
\begin{equation}
\hspace*{-0.1in}  \widehat{\theta}^{\mathrm{WIPW},j}_{t} = 
\left[ \sumiN 
\frac{
W_{\tau_i}^j\Delta_{t-1,i}C_{i}^j 
}{ 
\{ \prod_{k=2}^K \pi_{\tau_i,k}^j(\overline{\bX}_{k,i})\} 
\pi_{\tau_i, 1}^j(\bX_{1,i})  
} \right]^{-1} \!\!\! \sumiN \frac{W_{\tau_i}^j\Delta_{t-1,i}C_{i}^j Y_i}{ 
\{ \prod_{k=2}^K \pi_{\tau_i,k}^j(\overline{\bX}_{k,i})\} 
\pi_{\tau_i, 1}^j(\bX_{1,i})},
\label{eq:IPWrand}
\end{equation}
where $W_{\tau}^j$ is a weight depending on
$\mathfrak{D}_{\tau-1}$ discussed further in Section~\ref{s:theory}.
The denominator in each term in (\ref{eq:IPWrand}) can be interpreted
as the propensity for receiving treatment consistent with $\bd^j$
through all $K$ stages and depends on $r_{\tau}^j$,
$j=1,\ldots,m$.  For example, in the cancer pain SMART with $K=2$, let
$X_1$ denote a patient's baseline pain score and $X_{2,1}$ their pain
score at the end of stage 1, so that $\Xbar_2 = (X_1,X_{2,1})^T$, and
let $X_{2,2} = I(X_{2,1} \leq 0.7\, X_1)$ denote response status.  For
embedded regime $\bd^1$, under which a subject will receive $A_1=0$
and then $A_2=0$ if they respond to $A_1$ and $A_2=1$ otherwise,
$C^1 = I(A_1=0, X_{2,2}=1, A_2=0)+I(A_1=0,X_{2,2}=0,A_2=1)$.  Because
regimes 1--4 assign $A_1=0$,
$\pi_{\tau, 1}^1(\bX_1) = \sum^4_{j=1} r_{\tau}^j$, and because
regimes 1 and 2 assign $A_2=0$ to responders and regimes 1 and 4
assign $A_2=1$ to nonresponders,
$\pi_{\tau, 2}^1(\overline{\bX}_2) = \{
I(X_{2,2}=1)(r_{\tau}^1+r_{\tau}^2) + I(X_{2,2}=0)
(r_{\tau}^1+r_{\tau}^4)\}/\sum^4_{j=1} r_{\tau}^j$.

Usual AIPW estimators incorporate baseline and interim information
\citep[e.g., ][Section~6.4.4]{tsiatis2020dynamic} to gain efficiency
over IPW estimators.  Accordingly, we consider a class of weighted
AIPW (WAIPW) estimators.  Let $W_{\tau}^{A,j}$ be a weight depending
on $\mathfrak{D}_{\tau-1}$; and
$\overline{C}_{k}^j = I\{ \Abar_k = \dbar_k^j(\Xbar_k) \}$,
$k=1,\ldots,K$, with $\overline{C}_{0}^j \equiv 1$.  Estimators in the
class are of the form
\begin{equation}
  \begin{aligned}
    &\widehat{\theta}_{t}^{\mathrm{WAIPW},j} = \left( \sumiN
      W_{\tau_i}^{A,j}\Delta_{t-1,i} \right)^{-1}  \sumiN
W_{\tau_i}^{A,j}\Delta_{t-1,i}\left( \frac{C_{i}^j Y_i}{ \left\{ \prod_{k=2}^K
                  \pi_{\tau_i,k}^j(\overline{\bX}_{k,i})\right\}
                  \pi_{\tau_i, 1}^j(\bX_{1,i})  }
                  \vphantom{\frac{\overline{C}_{k,i}^j}{\left\{\prod_{v=2}^{k-1}
                  \pi_{d,v,tau_i} (\overline{\bX}_{v,i})\right\} 
                  \pi_{d,1\tau_i}(\bX_{1,i}) } }\right.
 \\
  &\,\,+\left.\sum_{k=1}^K \left[\frac{\overline{C}_{k-1,i}^j}{ 
  \left\{\prod_{v=2}^{k-1} \pi_{\tau_i,v}^j(\overline{\bX}_{v,i})\right\}
    \pi_{\tau_i,1}^j(\bX_{1,i}) } - 
  \frac{\overline{C}_{k,i}^j}{ 
  \left\{\prod_{v=2}^{k} \pi_{\tau_i,v}^j(\overline{\bX}_{v,i})\right\}
    \pi_{\tau_i,1}^j(\bX_{1,i}) } \right] L^j_k(\overline{\bX}_{k,i}) \right), 
    \label{eq:AIPWrand}
  \end{aligned}
  \end{equation}
where $L^j_k(\xbar_k)$ is an arbitrary function of $\xbar_k$ and
$\bd^j$, $k=1,\ldots,K$; and
$\prod_{v=1}^0 \pi_{\tau,v}^j (\bx_v) \equiv 1$.

The optimal choice is
$L^j_k(\xbar_k) = E\{ Y^*(\bd^j) | \Xbar_k=\xbar_k, \Abar_{k-1} =
\dbar_{k-1}^j(\xbar_{k-1})\}$, $k=1,\ldots,K$, which can be modeled
and estimated by adapting the backward iterative scheme in
\citet[Section~6.4.2]{tsiatis2020dynamic}.  Just as the
denominators of each term in (\ref{eq:IPWrand}) and
(\ref{eq:AIPWrand}) depend on 
$r_{\tau}^j$, $j=1,\ldots,m$, based on $\Data_{\tau-1} \subseteq
\Data_{t-1}$, the fitted
models used to approximate the optimal $L^j_k(\Xbar_k)$ should be based on
$\Data_{\tau-1}$.
At stage $K$, define
$Q_K(\xbar_K, \abar_K) = E(Y | \Xbar_K=\xbar_K, \Abar_K=\abar_K)$
and
$V_K^j(\xbar_K, \abar_{K-1}) = Q_K\{ \xbar_K, \abar_{K-1},
d_K^j(\xbar_K, \abar_{K-1}) \}$.  Posit a model
$Q_K(\xbar_K, \abar_K;\bbeta_K)$ for $Q_K(\xbar_K, \abar_K)$ indexed
by $\bbeta_K$, e.g., a linear or logistic model for continuous or
binary $Y$, respectively.  Let $\hatbeta_{K,t}$ denote an estimator
for $\bbeta_K$ based on subjects in
$\mathfrak{D}_{t-1}$ for whom $\Delta_{t-1}=1$, e.g., using least
squares or maximum likelihood, and define 
$\widehat{Q}_{K, t}(\xbar_K, \abar_K) = Q_K(\xbar_K,
\abar_K;\hatbeta_{K,t})$ and
$\tildeV_{K,t}^j(\xbar_K, \abar_{K-1}) =
\widehat{Q}_{K,t}\{ \xbar_K, \abar_{K-1}, d_K^j(a_K)\}$. Recursively for $k=K-1, \ldots, 1$, define
$Q_k^j(\xbar_k, \abar_k) = \allowbreak E\{V_{k+1}^j(\Xbar_{k+1}, \Abar_{k})|\Xbar_{k}=\xbar_k, \Abar_{k}=\abar_k \}$
and $V_k^j(\xbar_k, \abar_{k-1}) = Q_k^j\{ \xbar_k, \abar_{k-1}, d_k^j(\xbar_k, \abar_{k-1})\}$.  Posit a 
model $Q_k^j(\xbar_k, \abar_k;\bbeta_k^j)$ indexed by $\bbeta_k^j$, and obtain
estimator $\widehat{\bbeta}_{k,t}^j$ based on subjects in 
$\mathfrak{D}_{t-1}$ for whom $\Delta_{t-1}=1$ by an appropriate regression method 
using as the outcome the pseudo outcomes
$\tildeV_{k+1,t}^j(\Xbar_{k+1}, \Abar_{k}) = \widehat{Q}_{k+1,t}^j
\{\Xbar_{k+1}, \Abar_{k}, d_{k+1}^j(\Xbar_{k+1}, \Abar_{k})  \}$,
and let $\widehat{Q}_{k,t}^j(\xbar_k, \abar_k)=Q_k^j(\xbar_k, \abar_k;\hatbeta_{k,t}^j)$.
As in \citet[Section~6.4.2]{tsiatis2020dynamic}, these models may involve
separate expressions for responders and nonresponders; and, if at
stage $k$ there is only one treatment option for a subject's history,
$Y$ (if $k=K-1$) or $\tildeV^j_{k+2,t}(\Xbar_{k+2}, \Abar_{k+1})$
(if $k<K-1$) can be ``carried back'' in place of
$\tildeV_{k+1,t}^j(\Xbar_{k+1}, \Abar_{k})$.   For each $i$ with
$\Delta_{t-1,i}=1$ in (\ref{eq:AIPWrand}), substitute
$Q_K\{\Xbar_{K,i},\dbar_K^j(\Xbar_{K,i}); \hatbeta_{K,\tau_i}\}$ and
$Q^j_k\{\Xbar_{k,i},\dbar_k^j(\Xbar_{k,i}); \hatbeta^j_{k,\tau_i}\}$
for $L^j_K(\Xbar_{K,i})$ and $L^j_k(\Xbar_{k,i})$, $k=1,\ldots,K-1$, respectively.

As the basis for RAR, we propose using (\ref{eq:IPWrand}) or
(\ref{eq:AIPWrand}), with or without weights, to obtain estimators
$\widehat{\theta}_{t}^j$ for $\theta^j = \calV(\bd^j)$,
$j=1,\ldots,m$, based on $\Data_{t-1}$; and take the estimated
confidence distribution for $\btheta$ needed to obtain
$\widehat{\rho}_{t}^j$ to form $r_{t}^j$, $j=1,\ldots,m$, to be the
asymptotic normal distribution for 
$\hattheta_t=\{\widehat{\theta}_{t}^1,\ldots,
\widehat{\theta}_{t}^m\}^T$ following from M-estimation theory 
\citep[e.g., ][Section~6.4.4]{tsiatis2020dynamic}, with the weights
treated as fixed.  

Basing confidence distributions and thus randomization
probabilities at week $t$ on (\ref{eq:IPWrand}) or (\ref{eq:AIPWrand})
uses data only on subjects in $\Data_{t-1}$ who have completed the
trial, $\Delta_{t-1}=1$.  To exploit partial information on
subjects still progressing through the trial at $t$, with
$\Delta_{t-1}=0$, it is possible to develop a weighted version of the
interim AIPW (IAIPW) estimator of \citet*{Cole2022}; see 
Appendix A.
Simulations in Section~\ref{s:sims} show negligible gains in
performance over (\ref{eq:IPWrand}) or (\ref{eq:AIPWrand}).

\vspace*{-0.35in}

\subsection{Sequential randomization based on the optimal regime}
\label{ss:eachstage}

\vspace*{-0.0in}

We propose methods for obtaining randomization probabilities at week
$t$ based on $\Data_{t-1}$ to be used to assign treatments for
subjects requiring randomization at $t$ at any stage $k=1,\ldots,K$.
Because the set of feasible treatments $\Psi_k(\bh_k)$ for a subject
with history $\bh_k$ at stage $k$ may depend on $\bh_k$, as when the
sets of options for responders and nonresponders to previous treatment
are different, randomization probabilities may be history dependent.
The approach uses Q-learning for estimation of an optimal,
individualized regime
\citep[e.g.,][Section~7.4.1]{tsiatis2020dynamic}.  We present the
approach when $Y$ is continuous and linear models are used; extensions
to other outcomes and more flexible models are possible
\cite[][]{moodie2014q}.

Define the Q-functions
$Q_K(\xbar_K,\abar_K) = E(Y \mid \Xbar_K=\xbar_K, \Abar_K=\abar_K)$
and, for $k=K-1,\ldots,1$,
$Q_k(\xbar_k,\abar_k) = E\left\{ V_{k+1}(\xbar_k,\bX_{k+1},\abar_k)
  \mid \Xbar_k=\xbar_k, \Abar_k=\abar_k \right\}$, where
$V_k(\xbar_k,\abar_{k-1}) = \max_{a_k \in \Psi_k(\xbar_k,\abar_{k-1})}
Q_k(\xbar_k,\overline{a}_{k-1},a_k)$, $k=1,\ldots, K$.  Posit models
$Q_k(\xbar_k,\abar_k;\bbeta_k) = \bphi_k(\xbar_k,\abar_k)^T \bbeta_k$,
where $\bphi_k(\xbar_k,\abar_k)$ is a $p_k$-dimensional feature
function, 
$k=1,\ldots,K$.  Randomization probabilities are obtained via the
following backward algorithm.  At stage $K$, obtain $\hatbeta_{K,t}$,
the ordinary least squares (OLS) estimator based on subjects in
$\Data_{t-1}$ for whom $\Delta_{t-1}=1$ and its estimated covariance
matrix
$\hatSig_{K,t} = \hatsigma^2_{t,K} \big\{ \sumiN \Delta_{t-1,i}
\bphi_K(\Xbar_{K,i},\Abar_{K,i})
\bphi_K(\Xbar_{K,i},\Abar_{K,i})^T\big\}^{-1}$,
$\hatsigma^2_{K,t} = N_t^{-1}\sumiN \Delta_{t-1,i} \{Y_i
-\bphi_K(\Xbar_{K,i},\Abar_{K,i})^T\hatbeta_{K,t}\}^2$,
$N_t = \sumiN \Delta_{t-1,i}$.
Based on these results, obtain the estimated confidence distribution
$\hatP_{\bbeta_{K,t}}$ for $\bbeta_K$ as described below.
For $k=K-1,\ldots,1$, define pseudo outcomes 
$\tildeV_{k+1,t}(\bbeta_{k+1}) = \max_{a_{k+1} \in
  \Psi(\Xbar_{k+1},\Abar_{k})} Q_{k+1}( \Xbar_{k+1},\Abar_{k},
a_{k+1}; \bbeta_{k+1})$ 
for subjects in $\Data_{t-1}$ with $\Delta_{t-1}=1$. If a subject's history is such that
$\Psi_{k+1}(\Xbar_{k+1},\Abar_{k})$ comprises a single treatment option,
$\tildeV_{k+1,t}(\bbeta_{k+1})$ can be taken equal to the pseudo outcome 
at step $k+2$ or $Y$ if $k=K-1$.  Obtain an estimator
for $\bbeta_k$, $k = K-1,\ldots,1$, by OLS as
\begin{equation}
 \hatbeta_{k,t}(\hatbeta_{k+1,t}) = \mbox{argmin}_{\bbeta} \sumiN
\Delta_{t-1,i} \{ \tildeV_{k+1,t,i}(\widehat{\bbeta}_{k+1,t}) -
\bphi_k(\Xbar_{k,i},\Abar_{k,i})^T \bbeta \}^2,
\label{eq:hatbetak}
\end{equation}
and estimated covariance matrix $\hatSig_{k,t}(\hatbeta_{k+1,t}) = \hatsigma^2_{k,t} \big\{
\sumiN \Delta_{t-1,i}
\bphi_k(\Xbar_{k,i},\Abar_{k,i})\bphi_k(\Xbar_{k,i},\Abar_{k,i})^T\big\}^{-1}$,  
$\hatsigma^2_{k,t}(\widehat{\bbeta}_{k+1,t}) = N_{t}^{-1}\sumiN \Delta_{t-1,i}\{ \tildeV_{k+1,t,i}(\widehat{\bbeta}_{k+1,t}) -
\bphi_k(\Xbar_{k,i},\Abar_{k,i})^T \hatbeta_{k,t}(\hatbeta_{k+1,t})\}^2$.  Based on
these results, obtain an estimated confidence distribution for $\hatP_{\bbeta_{k,t}}$
as described next.  Note that, if subjects enter the SMART by a random
process, basing the fitted models on subjects in $\Data_{t-1}$ with
$\Delta_{t-1}$ is reasonable, as these subjects are representative of
the subject population.

Because $\hatbeta_{K,t}$ is a standard OLS estimator, it is natural to
approximate the confidence distribution $\hatP_{\bbeta_{K,t}}$ for
$\bbeta_K$ by $\calN(\hatbeta_{K,t}, \hatSig_{K,t} )$.  However, because
(\ref{eq:hatbetak}) is not a standard regression problem,
$\hatbeta_{k,t}(\hatbeta_{k+1,t})$, $k<K$, is a nonregular estimator
\citep[Section~10.4.1]{tsiatis2020dynamic}.  Thus, 
usual large sample theory does not apply and confidence
intervals based on (unadjusted) normal or bootstrap
approximations need not achieve nominal coverage.  Accordingly, we
obtain a confidence distribution $\hatP_{\bbeta_{k,t}}$ for
$\bbeta_k$, $k=1,\ldots,K-1$, in the spirit of a projection interval
\citep{laber_lizotte_qian_pelham_murphy_2014}, which faithfully
represents the uncertainty in
$\widehat{\bbeta}_{K-1,t},\dots,\widehat{\bbeta}_{1,t}$.

We demonstrate this approach for $K=3$.  For final stage 3, draw a
sample of size $B_3$, $\tilbeta_{3,t}^1,\ldots,\tilbeta_{3,t}^{B_3}$,
say, from $\hatP_{\bbeta_{3,t}}$, the approximate normal sampling
distribution $\calN(\hatbeta_{3,t}, \hatSig_{3,t})$ as above.
At stage 2, first draw a sample
$\tilbeta_{3,t}^1,\ldots,\tilbeta_{3,t}^{b_3}$ of size $b_3$ from
$\hatP_{\bbeta_3,t}$.  For each $\tilbeta_{3,t}^b$, form pseudo
outcomes
$\tildeV_{3,t,i}(\tilbeta_{3,t}^b) =
\max_{a_3 \in \Psi(\Xbar_{3,i},\Abar_{2,i})}
Q_3(\Xbar_{3,i},\Abar_{2,i},a_3; \tilbeta_{3,t}^b)$ and obtain
$\hatbeta_{2,t}(\tilbeta_{3,t}^b)$ by OLS analogous to
(\ref{eq:hatbetak}) with $k=2$, and obtain
$\hatSig_{2,t}(\tilbeta_{3,t}^b)$ similarly.  Then draw a sample
$\tilbeta_{2,t}^1,\ldots,\tilbeta_{2,t}^{b_2}$ from
$\calN\{ \hatbeta_{2,t}(\tilbeta_{3,t}^b),
\hatSig_{2,t}(\tilbeta_{3,t}^b)\}$.  The $b_3$ samples of size $b_2$
corresponding to each $\tilbeta_{3,t}^b$, $b=1,\ldots,b_3$,
collectively comprise a sample of size $B_2 = b_3 \times b_2$ from the
confidence distribution $\hatP_{\bbeta_{2,t}}$ for $\bbeta_2$.  At stage
1, again draw a sample $\tilbeta_{3,t}^1,\ldots,\tilbeta_{3,t}^{b_3}$
from $\hatP_{\bbeta_3,t}$ and obtain
$\hatbeta_{2,t}(\tilbeta_{3,t}^b)$ and
$\hatSig_{2,t}(\tilbeta_{3,t}^b)$, $b=1,\ldots,b_3$, as above.  Then
draw a sample $\tilbeta_{2,t}^1,\ldots,\tilbeta_{2,t}^{b_2}$ from
$\calN\{ \hatbeta_{2,t}(\tilbeta_{3,t}^b),
\hatSig_{2,t}(\tilbeta_{3,t}^b)\}$.  For each of $\tilbeta_{2,t}^b$,
$b=1,\ldots,b_2$, form pseudo outcomes
$\tildeV_{2,t,i}(\tilbeta_{2,t}^b) = \max_{a_2 \in \Psi(\Xbar_{2,i},A_{1,i})}
Q_2(\Xbar_{2,i},A_{1,i},a_2; \tilbeta_{2,t}^b)$, and obtain
$\hatbeta_{1,t}(\tilbeta_{2,t}^b)$ by OLS analogous to
(\ref{eq:hatbetak}) with $k=1$, and obtain
$\hatSig_{t,1}(\tilbeta_{2,t}^b)$.  Then draw a sample
$\tilbeta_{1,t}^1,\ldots,\tilbeta_{1,t}^{b_1}$ from
$\calN\{ \hatbeta_{1,t}(\tilbeta_{2,t}^b),
\hatSig_{1,t}(\tilbeta_{2,t}^b)\}$, $b=1,\ldots,b_2$.  The
$b_3 \times b_2$ samples of size $b_1$ corresponding to each
combination of draws $\tilbeta_{3,t}^b$, $\tilbeta_{2,t}^{b'}$,
$b = 1,\ldots,b_3$, $b' = 1,\ldots,b_2$, collectively comprise a
sample of size $B_1 = b_3 \times b_2 \times b_1$ from the confidence
distribution $\hatP_{\bbeta_1,t}$.  This procedure is embarrassingly
parallel, and as it involves only draws from a multivariate normal
distribution, it is not computationally burdensome.

Having generated $B_k$ draws from $\hatP_{\bbeta_{t,k}}$,
$k=1,\ldots,K$, at week $t$, we can obtain beliefs, i.e., draws from
approximate posterior distributions, which can be translated into
randomization probabilities for subjects requiring treatment
assignments at any stage $k=1,\ldots,K$ at week $t$.  Because the
$Q$-functions depend on patient history, the randomization
probabilities under TS can vary across subjects at a given time even
if their feasible treatments are the same.  To see this, suppose that
there are $s_k$ feasible sets of treatments at stage $k$, denoted by
$\zeta_{k}^u = \{a_{k}^{1,u},\ldots,a_{k}^{m_{k}^u,u}\}$,
$u=1,\ldots, s_k$.  Further, suppose that there are $\ell_{k,t}$
subjects requiring randomization at stage $k$ at week $t$ and that
$\ell_{k,t}^u$ are eligible for the $u$th feasible set.
For the $u$th feasible set, randomization probabilities can be
obtained for each subject $v=1,\ldots,\ell_{k,t}^u$ by defining the
belief for $a_{k}^{j,u}$, $j=1,\ldots,m_{k}^u$, depending on
$\Data_{t-1}$ for subject $v$ as
$\widehat{\rho}_{t,v}^{j,k,u} = B_k^{-1} \sum^{B_k}_{b=1} I\{ a_{k}^{j,u} =
\mbox{argmax}_{a \in \zeta_{k}^u} Q_k(\Xbar_{k,v},\Abar_{k-1,v},a;
\tilbeta_{k,t}^b)\},$ from which randomization probabilities for
subject $v$ are obtained.  If the Q-function model depends
on the history $\bh_k = (\xbar_k,\abar_{k-1})$ only through the
components such as previous treatment and response status that dictate
the feasible set, then this approach will yield the same probabilities
for all $\ell_{k,t}^u$ subjects.  Otherwise, randomization
probabilities will be individual-specific, depending on covariate and
treatment information in addition to $\Data_{t-1}$, which could be
logistically complex.  A second approach is to consider all
$(m_{k}^u)^{\ell_{k,t}^u}$ possible configurations for assigning the
options in $\zeta_{k}^u$ to the $\ell_{k,t}^u$ subjects and define the
beliefs and thus randomization probabilities based on the
configuration that maximizes the average $k$th Q-function across
subjects.  Letting $\Lambda_{k}^u$ be the set of all possible
configurations, writing the $j$th configuration as
$\mbox{\sf a}^j = (\mbox{\sf a}_1^j,\ldots,\mbox{\sf
  a}_{\ell_{k,t}^u}^j)$, define
    $\widehat{\rho}_{t}^{j,k,u} = B_k^{-1} \sum_{b=1}^{B_k}
  I\Big\{ \mbox{\sf a}^j = \mbox{argmax}_{\mbox{\sf a} \in \Lambda_{k}^u} 
(\ell_{k,t}^u)^{-1} \sum_{v=1}^{\ell_{k,t}^u}
Q_k(\Xbar_{k,v},\Abar_{k-1,v},\mbox{\sf a};
\tilbeta_{k,t}^b)\Big\}.$
Simulation experiments (not shown here) suggest that
the two approaches perform similarly.  
  
\vspace*{-0.25in}

\section{Post-Trial Inference}
\label{s:theory}

Although the potential outcomes $\bW^*_i$, $i=1,\ldots,N$, are
independent and identically distributed (i.i.d.), under any form of
RAR, the observed data are not, as the randomization probabilities are
functions of the past data $\Data_{t-1}$ \citep{ZhangMest2021}.  Thus,
post-trial evaluation of $\calV(\bd^1),\ldots,\calV(\bd^m)$ based on
the usual unweighted IPW or AIPW estimators is potentially
problematic, as standard asymptotic theory for these estimators, which
assumes i.i.d. data, does not apply
\citep[e.g.,][]{bibaut2021post,zhang2021inference}.  Thus, we
adapt the approach of \citet{ZhangMest2021} and choose the weights
in (\ref{eq:IPWrand}) and (\ref{eq:AIPWrand}) so that
asymptotic normality for these estimators can be established via the
martingale central limit theorem 
 We sketch the rationale; details are given in Appendix B.


To emphasize the key ideas, consider a simplified setting in which a
fixed number of subjects, $n$, enrolls at each week and
$T \rightarrow \infty$, so that the total number of subjects
$N = nT \rightarrow \infty$.  To simplify notation, take $n=1$.  At
the end of the trial, with all data complete, reindexing subjects by
$t=1,\ldots,T$, for regime $j$, the estimators
$\widehat{\theta}_{T}^{\mathrm{WIPW},j}$ and
$\widehat{\theta}_{T}^{\mathrm{WAIPW},j}$, say, are
solutions in $\theta^j$ to an estimating equation 
$\sum_{t=1}^T M_{t}^j (\Xbar_{K,t},\Abar_{K,t},Y_{t}; \theta^j) = 0$;
e.g., for (\ref{eq:IPWrand}),
\begin{equation}
M_{t}^j (\Xbar_{K,t},\Abar_{K,t},Y_t; \theta^j) = \frac{ W_{t}^j
  C^j_t}{ \{ \prod_{k=2}^K \pi_{t,k}^j(\Xbar_{k,t})\}
  \pi_{t,1}^j(\bX_{1,t}) } (Y_t - \theta^j). 
\label{eq:MestIPW}
\end{equation}
Critical to the proof 
is that the weights $W_{t}^j$ and $W^{A,j}_{t}$ are chosen
so that (i) the estimating equations remain conditionally unbiased, 
$E\{ M_{t}^j (\Xbar_{K,t},\Abar_{K,t},Y_t; \theta^j) | \Data_{t-1}\} = 0$,
and
(ii) the variance is stabilized, 
$E[ \{M_{t}^j (\Xbar_{K,t},\Abar_{K,t},Y_t;
\theta^j)\}^2 | \Data_{t-1}]$ $= \sigma^2 > 0$ for all $t$.  Writing
$\widehat{\theta}_{T}^j$ to denote either estimator and defining
$\delta_{T}^j = T^{-1} \sum_{t=1}^T  \partial/\partial \theta^j \{ M_{t}^j
(\Xbar_{K,t},\Abar_{K,t},Y; \theta^j )\}_{\theta^j=\widehat{\theta}_{T}^j}$ and
$(\sigma^{j}_T)^2 = T^{-1} \sum_{t=1}^T M_{t}^j(\Xbar_{K,t},\Abar_{K,t},Y_t;  \widehat{\theta}_{T}^j)^2 $, we
argue in Appendix B that, under standard regularity conditions,  
$\widehat{\theta}_{T,j} \inp \theta^j$ and
\begin{equation}
\delta^j_{T} (\sigma_T^j)^{-1} T^{1/2} (\widehat{\theta}_{T}^j - \theta^j) \inD \calN(0,1).
\label{eq:normality}
\end{equation}
In Section~\ref{s:sims}, we use (\ref{eq:normality}) to construct confidence intervals and
bounds for $\theta^j$, $j=1,\ldots,m$.

We sketch arguments for the WIPW estimator to show that (i) holds and
how to choose the weights to guarantee (ii); see 
Appendix B for details and arguments for the WAIPW estimator.
Define
$\Xbar_k^*(\bd^j) = \{\bX_1, \bX_2^*(\bd^j), \ldots, \bX_k^*(\bd^j)\}$, $k=1,\ldots,K$, 
and write
$\pi_t^{*,j}\{ \Xbar_K^*(\bd^j) \} =  \prod_{k=2}^K\pi_{t,k}^j\{\Xbar_k^*(\bd^j)\} \pi_{t,1}^j(\bX_1)$.
When $C^j_t=1$, $Y_t = Y^*(\bd^j)$, $\Xbar_{K,t} = \Xbar_K^*(\bd^j)$, so
that (i) is 
\begin{align*}
  E&\{ M_{t}^j (\Xbar_{K,t},\Abar_{K,t},Y_t; \theta^j)  | \Data_{t-1}\}
   =  W_{t}^j E \left[ \left. \frac{ C^j_t}{ \pi_t^{*,j}\{ \Xbar_K^*(\bd^j) \} } \{Y^*(\bd^j) - \theta^j\} \right|
     \Data_{t-1}\right] \\
  &= W_{t}^j  E \left[ \left. \frac{ E(C^j_t| \mathcal{W}^*, \Data_{t-1})}{ \pi_t^{*,j}\{ \Xbar_K^*(\bd^j) \}
   } \{Y^*(\bd^j) - \theta^j\} \right|
     \Data_{t-1}\right] = W_{t}^j E\{ Y^*(\bd^j) - \theta^j \} = 0,
  \end{align*}
 because, as in \citet[][Section~6.4.3]{tsiatis2020dynamic}, $E(C^j_t|
 \mathcal{W}^*, \Data_{t-1}) = \pi_\tau^{*,j}\{ \Xbar_K^*(\bd^j)\}$, and $\mathcal{W}^*$ and thus $Y^*(\bd^j)$ is
 independent of $\Data_{t-1}$.  Using similar manipulations, 
 \begin{equation}
E\{ M_{t}^j (\Xbar_{K,t},\Abar_{K,t},Y_t; \theta^j)^2  | \Data_{t-1}\} =
    (W_{t}^j )^2 E\left[ \left. \frac{ \{ Y^*(\bd^j) - \theta^j \}^2
      }{\pi_t^{*,j}\{ \Xbar_K^*(\bd^j) \}} \right| \Data_{t-1} \right].
      \label{eq:Mvar}
\end{equation}
Thus, to ensure (ii), $W_{t}^j$ should be chosen so that
(\ref{eq:Mvar}) is a constant depending only on $j$.

We demonstrate the choice of $W_{t}^j$ in practice for regime 1 of
the cancer pain SMART, $K=2$, and up-front randomization.  From
Section~\ref{ss:embedded}, letting
$X^*_{2,2}(\bd^1) = I\{ X^*_{2,1}(\bd^1) \leq 0.7\, X_1\}$,
$\pi_{t, 1}^1(\bX_1) = \pi_{t,1}^1 = \sum^4_{j=1} r_{t}^j$,
and
$\pi_{t, 2}^1\{\Xbar^*_2(\bd^1)\}= I\{ X_{2,2}^*(\bd^1) = 1\}
\pi^{1(1)}_{t,2} + I\{ X_{2,2}^*(\bd^1) = 0\} \pi^{1(0)}_{t,2}$,
where $\pi^{1 (1)}_{t,2}= (r_{t}^1+r_{t}^2)/\pi_{t,1}^1$
and $\pi^{1 (0)}_{t,2}(r_{t}^1+r_{t}^4)/\pi_{t,1}^1$.
Then, using $\mathcal{W}^* \independent\Data_{t-1}$, (\ref{eq:Mvar}) is
  \begin{equation}
(W_{t,1} )^2 \left( \frac{ \mu^{1(1)}} {\pi^{1
      (1)}_{t,2}\pi_{t, 1}^1} + \frac{ \mu^{(0)}_1} {\pi^{1
      (0)}_{t,2}\pi_{t, 1}^1} \right), \,\,\, \mu^{1(s)} = E[  I\{X_{2,2}^*(\bd^1) = s\} \{ Y^*(\bd^1) - \theta^1\}^2], \,s = 0, 1.
\label{eq:regime1var}
\end{equation}
Then, in the original notation, estimate
$\mu^{1(s)}_1$ at week $t$ by
$\hatmu^{1(s)}_{t} = N^{-1}_t \sumiN \big[ \Delta_{t-1,i}I(X_{2,2,i}=s) C^1_i (Y_i-\widetilde{\theta}_{t}^1)^2/\{\pi^{1
    (s)}_{\tau_i,2}\pi_{\tau_i, 1}^1\}\big]$, $s=0,1$, based on
  $\Data_{t-1}$, where
$\widetilde{\theta}_{t}^1$ is an estimator for $\theta^1$ using
$\Data_{t-1}$ (we use the unweighted IPW estimator).  Setting
(\ref{eq:regime1var}) equal to a constant $\Xi^1$, say, and defining
$\widehat{\Xi}_{t}^1 = \sum_{s=0}^1 \{\hatmu^{1(s)}_{t}/(\pi^{1(s)}_{t,2}\pi_{t, 1}^1) \}$,
estimate $\Xi^1$ based on the burn-in data, for which
$W_{t}^j \equiv 1$ and $r_{t}^j$, $j=1,\ldots,8$, are fixed
(nonadaptive) for $t \leq t^*$, by $\widehat{\Xi}_{t^*}^1$.  Then
  for $t \geq t^*+1$, take $W_{t}^1 =
(\widehat{\Xi}^1_{t^*}/\widehat{\Xi}^1_{t})^{1/2}$.
See Appendix B for considerations for sequential RAR.


\section{Simulation studies}
\label{s:sims}


\subsection{Up-front randomization among embedded regimes}
\label{ss:embeddedsims}


We present results of a simulation study involving 5000 Monte Carlo
trials under a scenario mimicking the cancer pain SMART introduced in
Section~\ref{s:intro} in which subjects are randomized using various
forms of up-front RAR as in Section~\ref{ss:embedded}.  
Each trial enrolls $N=1000$ subjects, with enrollment times uniform
over (integer) weeks 1-24. Upon enrollment, we draw baseline pain
score $X_1 \sim \calN(5,1)$ and assign stage 1 treatment
$A_1 \in \mathcal{A}_1=\{0,1\}$.  Six weeks after $A_1$ is assigned,
second-stage pain score is generated as
$X_{2,1} = \gamma_{1,0} + \gamma_{1,1} X_1 + \gamma_{1,2} A_1 +
\varepsilon_1$, where $\varepsilon_1 \sim \calN(0,1)$, and response
status after the first stage is $X_{2,2}=I(X_{2,1}<0.7X_1)$, which,
with $A_1$, dictates the feasible subset of $\calA_2 = \{0,1,2,3,4\}$
from which stage 2 treatment $A_2$ is assigned.  Six weeks later, the
outcome is generated as
$Y=\gamma_{2,0} + \gamma_{2,1} X_1 + \gamma_{2,2} A_1 + \gamma_{2,3}
X_{2,1} + \gamma_{2,4}I(A_2=1) + \gamma_{2,5}I(A_2=2) +
\gamma_{2,6}I(A_2=3) + \gamma_{2,7}I(A_2=4) + \varepsilon_2$, where
$\varepsilon_2 \sim \calN(0,1)$.  With
$\bgamma_1 = (\gamma_{1,0},\gamma_{1,1},\gamma_{1,2})^T = (0.00, 0.90,
-1.50)^T$ and
$\bgamma_2 = (\gamma_{2,0},\ldots,\gamma_{2,7}) =
(0.00,0.30,-0.75,0.60,-0.25,-0.75,-0.75,-0.85)^T$, for the 
$m=8$ embedded regimes defined in Figure~\ref{fig:example_smart}, 
$\{ \calV(\bd^1),\ldots,\calV(\bd^8)\} = (\theta_1,\ldots,\theta_8) =
(-0.126, -0.374, -0.500, -0.251, -2.408,$ $-2.401, -2.494, -2.501)$, so
that regime 8 is optimal, as larger reductions in pain are desirable.
The burn-in period ends at the time $t^*$ when each of the $m=8$ regimes has at
least 25 subjects who have completed the trial with experience
consistent with following the regime.

We compare the performance of up-front RAR using TS with $c_t = 0.5$
and 1 for all $t$ based on the WIPW and WAIPW estimators
(\ref{eq:IPWrand}) and (\ref{eq:AIPWrand}) and the unweighted IAIPW
estimator to that of simple, uniform randomization (SR) among the
eight embedded regimes. 
For each TS variant, at any $t$ we
impose clipping constants of 0.05 and 0.95, so if $r_{t}^j<0.05$ for any
$j=1,\ldots,8$, set $r_{t}^j=0.05$ (and likewise for 0.95), and then normalize $r_{t}^j$,
$j=1,\ldots,8$, to sum to one.  In the WAIPW and IAIPW estimators, 
$Q_{2}(\xbar_2,\abar_2;\bbeta_2) = \beta_{2,0} + \beta_{2,1}x_1 + \beta_{2,2}I(a_1=1) +
\beta_{2,3}x_{2,1} + \beta_{2,4}I(a_2=1) + \beta_{2,5}I(a_2=2) +
\beta_{2,6}I(a_2=3) + \beta_{2,7}I(a_2=4)$ and
$Q_{1}^j(x_1,a_1;\bbeta_{1}^j) = \beta_{1,0}^j + \beta_{1,1}^j x_1 + 
\beta_{1,2}^ja_1 + \beta_{1,3}^j x_1a_1$.

Results are displayed in Table~\ref{t:one}.  To evaluate how TS
improves in-trial outcomes, we report Monte Carlo average outcome
across individuals in the trial; average proportion of subjects
assigned $A_1=1$, the optimal first-stage treatment; and average
proportion of subjects who were consistent with the optimal regime
within the trial.  Outcomes are improved using TS over SR, with TS
resulting in lower average outcome and higher rates of assigning
optimal stage 1 treatment and optimal regime.  More aggressive
adaptation, $c_t=1$, yields more favorable results than $c_t=0.5$ when
using any estimator. Results are more favorable for WAIPW- and
IAIPW-based randomization, as those estimators exploit covariate
information, with the gains mostly at stage 2. Using weighted
vs. unweighted estimators yields modest in-trial improvements.



\begin{table}
  \centering
  \caption{Simulation results using up-front RAR based on TS for 5000 Monte Carlo (MC) replications
    for the scenario in Section~\ref{ss:embeddedsims}. Columns
    indicate the randomization method: WAIPW(0.5) is TS 
    based on the WAIPW estimator (\ref{eq:AIPWrand}) with $c_t=0.5$
    for all $t$ and AIPW(1) uses $c_t=1$ for all $t$; WIPW($c_t$) and
    IAIPW($c_t$) are defined similarly. SR denotes simple, uniform
    randomization. Mean $Y$ denotes the MC average mean outcome for
    the 1000 individuals in the trial; lower mean outcomes are more
    favorable. Proportion $A_1$ Opt is the MC average proportion of
    subjects assigned the optimal treatment at the first-stage;
    Proportion Regime Opt is the MC average proportion of subjects who
    followed the optimal regime in the trial. For estimation results,
    Regime Correct denotes the proportion of trials we correctly
    estimate $\bd^8$ to be the optimal regime; $\mathcal{V}(\bd^8)$
    MSE is the MC mean squared error for regime
    8. For the estimation results, the term in parentheses, e.g.,
    (IPW), denotes the estimator used after the trial is completed.
    $\mathcal{V}(\bd^8)$ $95\%$ CI is the MC proportion of $95\%$
    confidence intervals that cover the true value; the term in
    parentheses is the estimator used to construct the confidence
    interval.  $\mathcal{V}(\bd^8)$ $95\%$ LB for lower confidence
    bounds $\mathcal{V}(\bd^8)$ $95\%$ UB for upper confidence bounds
    are defined similarly. Bold values indicate the most favorable
    result among the randomization methods.  Standard deviations of
    entries are in parentheses}
\label{t:one}
 \begin{center}
\scalebox{0.65}{
 \begin{tabular}{l c c c c c c c} \Hline
& SR & WIPW(0.5) & WIPW(1) & WAIPW(0.5) & WAIPW(1) & IAIPW(0.5) & IAIPW(1) \\
\hline
\underline{In Trial} &  & & & & & &  \\
Mean $Y$ & -1.380 (0.001) & -1.795 (0.001) & -1.992 (0.001) & -1.794 (0.001) & \textbf{-1.996} (0.001) & -1.793 (0.001) & -1.993 (0.001) \\
Proportion $A_1$ Opt & 0.500 (0.000) & 0.691 (0.000) & \textbf{0.782} (0.000) & 0.691 (0.000) & \textbf{0.782} (0.000) & 0.690 (0.000) & \textbf{0.782} (0.000) \\
Proportion Regime Opt & 0.250 (0.000) & 0.390 (0.001) & 0.470 (0.001) & 0.401 (0.001) & \textbf{0.498} (0.002) & 0.402 (0.001) & 0.491 (0.001) \\
 &  &  &  & & &  &  \\
\underline{Estimation} &  & & & & & &  \\
Regime Correct (IPW) & 0.433 (0.007) & \textbf{0.462} (0.007) & 0.444 (0.007) & 0.449 (0.007) & 0.409 (0.007) & 0.441 (0.007) & 0.430 (0.007) \\
Regime Correct (WIPW) & 0.397 (0.007) & \textbf{0.468} (0.007) & 0.460 (0.007) & 0.463 (0.007) & 0.436 (0.007) & 0.452 (0.007) & 0.460 (0.007) \\
Regime Correct (AIPW) & 0.529 (0.007) & \textbf{0.562} (0.007) & 0.534 (0.007) & 0.559 (0.007) & 0.523 (0.007) & 0.555 (0.007) & 0.534 (0.007) \\
Regime Correct (WAIPW) & 0.516 (0.007) & 0.548 (0.007) & 0.540 (0.007) & \textbf{0.555} (0.007) & 0.524 (0.007) & 0.545 (0.007) & 0.539 (0.007) \\
 &  &  &  & & &  &  \\
$\mathcal{V}(\bd^8)$ MSE $\times 10^2$ (IPW) & 0.814 (0.017) & 0.619 (0.013) & 0.735 (0.024) & \textbf{0.591} (0.013) & 0.671 (0.019) & 0.619 (0.013) & 0.669 (0.016)  \\

$\mathcal{V}(\bd^8)$ MSE $\times10^2$ (WIPW) & 1.185 (0.024) & 0.598 (0.013) & 0.646 (0.012) & \textbf{0.560} (0.012) & 0.576 (0.015) & 0.578 (0.012) & 0.574 (0.014) \\

$\mathcal{V}(\bd^8)$ MSE $\times 10^2$ (AIPW) & 0.544 (0.011) & 0.419 (0.008) & 0.463 (0.012) & \textbf{0.407} (0.008) & 0.446 (0.011) & 0.438 (0.009) & 0.432 (0.009) \\

$\mathcal{V}(\bd^8)$ MSE $\times10^2$ (WAIPW) & 0.552 (0.012) & 0.411 (0.008) & 0.419 (0.010) & \textbf{0.400} (0.008) & 0.420 (0.011) & 0.427 (0.009) & 0.411 (0.009) \\
 &  &  &  & & &  &  \\
\underline{Coverage} &  & & & & & &  \\
$\mathcal{V}(\bd^8)$ 95\% CI (IPW) & 0.949 (0.003) & 0.950 (0.003) & 0.949 (0.003) & 0.954 (0.003) & 0.947 (0.003) & 0.951 (0.003) & 0.950 (0.003)  \\
$\mathcal{V}(\bd^8)$ 95\% CI (WIPW) & 0.943 (0.003) & 0.948 (0.003) & 0.945 (0.003) & 0.953 (0.003) & 0.944 (0.003) & 0.949 (0.003) & 0.948 (0.003) \\
$\mathcal{V}(\bd^8)$ 95\% CI (AIPW) & 0.944 (0.003) & 0.952 (0.003) & 0.951 (0.003) & 0.952 (0.003) & 0.954 (0.003) & 0.950 (0.003) & 0.950 (0.003) \\
$\mathcal{V}(\bd^8)$ 95\% CI (WAIPW) & 0.947 (0.003) & 0.951 (0.003) & 0.950 (0.003) & 0.951 (0.003) & 0.948 (0.003) & 0.947 (0.003) & 0.944 (0.003) \\
 &  &  &  & & &  &  \\
$\mathcal{V}(\bd^8)$ 95\% LB (IPW) & 0.952 (0.003) & 0.960 (0.003) & 0.960 (0.003) & 0.954 (0.003) & 0.954 (0.003) & 0.950 (0.003) & 0.953 (0.003)  \\
$\mathcal{V}(\bd^8)$ 95\% LB (WIPW) & 0.946 (0.003) & 0.954 (0.003) & 0.948 (0.003) & 0.950 (0.003) & 0.947 (0.003) & 0.947 (0.003) & 0.945 (0.003) \\
$\mathcal{V}(\bd^8)$ 95\% LB (AIPW) & 0.944 (0.003) & 0.959 (0.003) & 0.955 (0.003) & 0.951 (0.003) & 0.954 (0.003) & 0.950 (0.003) & 0.952 (0.003) \\
$\mathcal{V}(\bd^8)$ 95\% LB (WAIPW) & 0.946 (0.003) & 0.957 (0.003) & 0.950 (0.003) & 0.947 (0.003) & 0.947 (0.003) & 0.944 (0.003) & 0.945 (0.003) \\
 &  &  &  & & &  &  \\
$\mathcal{V}(\bd^8)$ 95\% UB (IPW) & 0.950 (0.003) & 0.944 (0.003) & 0.940 (0.003) & 0.955 (0.003) & 0.942 (0.003) & 0.949 (0.003) & 0.944 (0.003)  \\
$\mathcal{V}(\bd^8)$ 95\% UB (WIPW) & 0.946 (0.003) & 0.947 (0.003) & 0.944 (0.003) & 0.955 (0.003) & 0.946 (0.003) & 0.952 (0.003) & 0.947 (0.003) \\
$\mathcal{V}(\bd^8)$ 95\% UB (AIPW) & 0.952 (0.003) & 0.947 (0.003) & 0.948 (0.003) & 0.954 (0.003) & 0.949 (0.003) & 0.944 (0.003) & 0.948 (0.003) \\
$\mathcal{V}(\bd^8)$ 95\% UB (WAIPW) & 0.952 (0.003) & 0.946 (0.003) & 0.952 (0.003) & 0.956 (0.003) & 0.949 (0.003) & 0.948 (0.003) & 0.950 (0.003) \\
\hline
\end{tabular}
}
\end{center}
\end{table}

To evaluate post-trial performance, Table~\ref{t:one} shows the
proportion of trials where $\bd^8$ is correctly identified as optimal
and the mean-squared error (MSE) for $\widehat{\theta}^8$ based on the
final data at the end of the trial.  Relative to SR, using RAR based
on any of the estimators generally identifies optimal regimes at
higher rates.
Performance of 95\% confidence intervals and lower/upper confidence
bounds for the true value based on the asymptotic theory in
Section~\ref{s:theory} is presented for $\calV(\bd^8)$; that for other
regimes is qualitatively similar.  Regardless of randomization
scheme, the nominal level is achieved in almost every case.
Figure~\ref{f:ftwo} shows the distribution of the 5000 centered and
scaled final value estimates $\widehat{\theta}^j$, $j=1, 8$, obtained
using the IPW, AIPW, WIPW, and WAIPW estimators based on RAR using the
WAIPW estimator with $c_t=1$; those for other schemes are similar.  As
found by \citet{ZhangMest2021} and others, unweighted
estimators result in mildly skewed distributions, while 
weighted estimators yield approximate standard normality.
Overall, the results suggest that the weighted estimators effect the
desired adjustment to yield asymptotic normality.  These observations
hold over all simulation scenarios we have tried with continuous and
binary outcomes, including where several regimes have the same
or similar values, under which departures from normality may be
particularly problematic \citep[e.g., ][]{Hadad2021}.

\begin{figure}
  \centering
	\includegraphics[width=\textwidth]{waipw.png}
  \caption{Monte Carlo distributions of centered and scaled estimates
    as in the theory of Section~\ref{s:theory} for selected estimators
    in the simulation in Section~\ref{ss:embeddedsims}.  The histograms
    correspond to the indicated estimator for $\mathcal{V}(\bd^1)$ and
    $\mathcal{V}(\bd^8)$ under up-front randomization using TS based on the
    WAIPW estimator (\ref{eq:AIPWrand}) with $c_t=1$ for all $t$.  The
    vertical line indicates mean zero, and the density of a standard
    normal distribution is superimposed.}
  \label{f:ftwo}
\end{figure}

\subsection{Sequential randomization at each stage based on optimal regimes}
\label{ss:qlearnsims}

\vspace*{-0.0in}

We report on a simulation study involving 5000 Monte Carlo trials
under the two-stage scenario in Section~\ref{ss:embeddedsims} to
compare the Q-learning-based sequential RAR approach using TS in
Section~\ref{ss:eachstage} with $c_t = 0.25, 0.50, 0.75$ and 1 to SR
and two SMART-AR methods proposed by
\citet{cheung_chakraborty_davidson_2014}, a conservatively-tuned
(AR-1) and a more aggressive (AR-2) version.  
.  To implement all RAR methods, we posit
linear models
$Q_{2}(\xbar_2,\abar_2;\bbeta_2) = \beta_{2,0} + \beta_{2,1}x_1 + \beta_{2,2}I(a_1=1) +
\beta_{2,3}x_{2,1} + \beta_{2,4}I(a_2=1) + \beta_{2,5}I(a_2=2) +
\beta_{2,6}I(a_2=3) + \beta_{2,7}I(a_2=4)$ and
$Q_{1}(x_1,a_1;\bbeta_{1}) = \beta_{1,0} + \beta_{1,1} x_1 + 
\beta_{1,2}a_1$.  For sequential RAR, we set
$B_1 = b_2 \times b_1 = 32 \times 32 = 1024$ and $B_2 = 1000$. The
tuning parameters for AR-1 are $b=10, \tau=0.5$ and
$b=100, \tau=0.025$ for AR-2 and $\lambda_t = t^{-1}\tau^{(1-b)}$ for
both; see \citet{cheung_chakraborty_davidson_2014}.  Clipping
constants of 0.05 and 0.95 were imposed on all methods.

At each week, for each RAR method, newly-enrolled subjects are
assigned stage 1 treatment using the same randomization
probability. Already-enrolled subjects who have reached stage 2 at
this week and require stage 2 randomization are partitioned into four
groups based on $(a_1, x_{2,2}) = (0,0), (0,1), (1,0), (1,1)$. Within
each group, randomization probabilities are calculated; thus,
second-stage probabilities are specific to each stage 1
treatment-response status combination. For both AR methods, we use the
sample average of $X_1$ to calculate the components that make up
the stage 2 probabilities; see
\citet{cheung_chakraborty_davidson_2014}.  The burn-in period ends at
the time $t^*$ when at least 25 subjects have completed the trial with
experience consistent with each of the $m=8$ embedded regimes.

Table~\ref{t:two} presents the results.  As in
Section~\ref{ss:embeddedsims}, the sequential RAR method results in
improved (over SR) in-trial outcomes on average by assigning
optimal treatments and regimes at higher rates. The AR methods also
improve on SR but are relatively conservative.  We were unable to
achieve more favorable results with differently-tuned versions of the
AR method.  For post-trial estimation, less-aggressive RAR identifies
the optimal regime at higher rates than SR.
The proposed method with $c_t=1$ and 0.75 yields the best post-trial
performance using any of the IPW, WIPW, AIPW, or WAIPW estimators but
lower rate of identifying the optimal regime. As expected, the AIPW
and WAIPW estimators are more efficient than the IPW and WIPW
estimators. While the primary goal of the weighted estimators is to
attain nominal coverage, an additional benefit is higher rates of
identifying the optimal regime.

\begin{table}
  \centering
  \caption{Simulation results using sequential RAR based on TS for 5000 Monte Carlo replications for
    the scenario in Section~\ref{ss:embeddedsims}. Columns indicate
    the randomization method: TS(0.25) is TS via
    Q-learning with $c_t=0.25$ for all $t$, TS(0.50), TS(0.75), and
    TS(1) are defined similarly. AR-1 is the conservatively-tuned
    version of SMART-AR and AR-2 is the more aggressive version. SR
    denotes simple, uniform randomization.  All entries are defined as
    in Table~\ref{t:one}.  }
  \label{t:two}
  \begin{center}
\scalebox{0.65}{
\begin{tabular}{l c c c c c c c}
  \hline
& SR & TS(0.25) & TS(0.50) & TS(0.75) & TS(1) & AR-1 & AR-2 \\
\hline
\underline{In Trial} &  & & & & & &  \\
Mean $Y$ & -1.380 (0.001) & -1.976 (0.001) & -1.999 (0.001) & -2.014 (0.001) & \textbf{-2.206} (0.001) & -1.950 (0.001) & -1.957 (0.001) \\
Proportion $A_1$ Opt & 0.500 (0.000) & 0.772 (0.001) & 0.780 (0.001) & 0.785 (0.001) & \textbf{0.790} (0.000) & 0.775 (0.000) & 0.775 (0.000) \\
Proportion Regime Opt & 0.250 (0.000) & 0.445 (0.001) & 0.491 (0.001) & 0.517 (0.001) & \textbf{0.538} (0.002) & 0.320 (0.001) & 0.351 (0.001) \\
 &  &  &  & & &  &  \\
\underline{Estimation} &  & & & & & &  \\
Regime Correct (IPW) & 0.433 (0.007) & 0.471 (0.007) & 0.463 (0.007) & 0.423 (0.007) & 0.404 (0.007) & 0.468 (0.007) & \textbf{0.484} (0.007) \\
Regime Correct (WIPW) & 0.397 (0.007) & 0.480 (0.007) & 0.471 (0.007) & 0.452 (0.007) & 0.434 (0.007) & 0.480 (0.007) & \textbf{0.488} (0.007) \\
Regime Correct (AIPW) & 0.529 (0.007) & 0.568 (0.007) & 0.549 (0.007) & 0.512 (0.007) & 0.484 (0.007) & 0.546 (0.007) & \textbf{0.569} (0.007) \\
Regime Correct (WAIPW) & 0.516 (0.007) & 0.569 (0.007) & 0.548 (0.007) & 0.530 (0.007) & 0.519 (0.007) & 0.558 (0.007) & \textbf{0.582} (0.007) \\
 &  &  &  & & &  &  \\
$\mathcal{V}(\bd^8)$ MSE $\times 10^2$ (IPW) & 0.814 (0.017) & 0.556 (0.012) & \textbf{0.541} (0.012) & 0.586 (0.015) & 0.660 (0.020) & 0.837 (0.018) & 0.676 (0.014)  \\
$\mathcal{V}(\bd^8)$ MSE $\times 10^2$ (WIPW) & 1.185 (0.024) & 0.504 (0.011) & \textbf{0.480} (0.011) & 0.515 (0.014) & 0.517 (0.014) & 0.767 (0.016) & 0.636 (0.013)  \\
$\mathcal{V}(\bd^8)$ MSE $\times 10^2$ (AIPW) & 0.544 (0.011) & 0.396 (0.009) & \textbf{0.383} (0.008) & 0.401 (0.010) & 0.433 (0.013) & 0.521 (0.011) & 0.442 (0.009)  \\
$\mathcal{V}(\bd^8)$ MSE $\times 10^2$ (WAIPW) & 0.552 (0.012) & 0.377 (0.008) & \textbf{0.368} (0.008) & 0.381 (0.010) & 0.378 (0.010) & 0.501 (0.010) & 0.432 (0.009) \\
 &  &  &  & & &  &  \\
\underline{Coverage} &  & & & & & &  \\
$\mathcal{V}(\bd^8)$ 95\% CI (IPW) & 0.949 (0.003) & 0.940 (0.003) & 0.951 (0.003) & 0.953 (0.003) & 0.954 (0.003) & 0.949 (0.003) & 0.948 (0.003)  \\
$\mathcal{V}(\bd^8)$ 95\% CI (WIPW) & 0.943 (0.003) & 0.940 (0.003) & 0.951 (0.003) & 0.952 (0.003) & 0.954 (0.003) & 0.949 (0.003) & 0.947 (0.003) \\
$\mathcal{V}(\bd^8)$ 95\% CI (AIPW) & 0.944 (0.003) & 0.947 (0.003) & 0.950 (0.003) & 0.958 (0.003) & 0.954 (0.003) & 0.950 (0.003) & 0.952 (0.003) \\
$\mathcal{V}(\bd^8)$ 95\% CI (WAIPW) & 0.947 (0.003) & 0.942 (0.003) & 0.947 (0.003) & 0.951 (0.003) & 0.952 (0.003) & 0.948 (0.003) & 0.951 (0.003) \\
 &  &  &  & & &  &  \\
$\mathcal{V}(\bd^8)$ 95\% LB (IPW) & 0.952 (0.003) & 0.948 (0.003) & 0.952 (0.003) & 0.955 (0.003) & 0.960 (0.003) & 0.940 (0.003) & 0.948 (0.003)  \\
$\mathcal{V}(\bd^8)$ 95\% LB (WIPW) & 0.946 (0.003) & 0.946 (0.003) & 0.948 (0.003) & 0.950 (0.003) & 0.954 (0.003) & 0.946 (0.003) & 0.949 (0.003) \\
$\mathcal{V}(\bd^8)$ 95\% LB (AIPW) & 0.944 (0.003) & 0.948 (0.003) & 0.948 (0.003) & 0.955 (0.003) & 0.957 (0.003) & 0.944 (0.003) & 0.949 (0.003) \\
$\mathcal{V}(\bd^8)$ 95\% LB (WAIPW) & 0.946 (0.003) & 0.946 (0.003) & 0.945 (0.003) & 0.947 (0.003) & 0.950 (0.003) & 0.946 (0.003) & 0.951 (0.003) \\
 &  &  &  & & &  &  \\
$\mathcal{V}(\bd^8)$ 95\% UB (IPW) & 0.950 (0.003) & 0.945 (0.003) & 0.950 (0.003) & 0.950 (0.003) & 0.946 (0.003) & 0.952 (0.003) & 0.950 (0.003)  \\
$\mathcal{V}(\bd^8)$ 95\% UB (WIPW) & 0.946 (0.003) & 0.947 (0.003) & 0.951 (0.003) & 0.951 (0.003) & 0.952 (0.003) & 0.949 (0.003) & 0.948 (0.003) \\
$\mathcal{V}(\bd^8)$ 95\% UB (AIPW) & 0.952 (0.003) & 0.946 (0.003) & 0.951 (0.003) & 0.952 (0.003) & 0.952 (0.003) & 0.954 (0.003) & 0.951 (0.003) \\
$\mathcal{V}(\bd^8)$ 95\% UB (WAIPW) & 0.952 (0.003) & 0.949 (0.003) & 0.953 (0.003) & 0.957 (0.003) & 0.956 (0.003) & 0.952 (0.003) & 0.947 (0.003) \\
\hline
\end{tabular}
}
\end{center}
\end{table}

Figure~\ref{f:ffour} shows the distributions of the 5000 centered and
scaled final value estimates $\widehat{\theta}^j$, $j=1, 8$, based on
sequential RAR using TS with $c_t=1$ (plots for other methods and
regimes are similar). For $\mathcal{V}(\bd^8)$, the weighted
estimators are approximately normal while the unweighted estimators
are slightly left skewed; for $\mathcal{V}(\bd^1)$, the unweighted
estimators are nonnormal, and the weighted estimators are
improved, although the variance is larger than one. We attribute this
behavior to undersampling of $\mathbf{d}^1$, the least effective
regime; the issue is less pronounced for less aggressive
randomization.
Coverage of confidence intervals and bounds for $\calV(\bd^8)$ is
mildly improved for the weighted estimators. Because of the left skew,
there is notable upper confidence bound undercoverage and lower bound
overcoverage for the unweighted estimators. 

\begin{figure}
  \centering
\includegraphics[width=\textwidth]{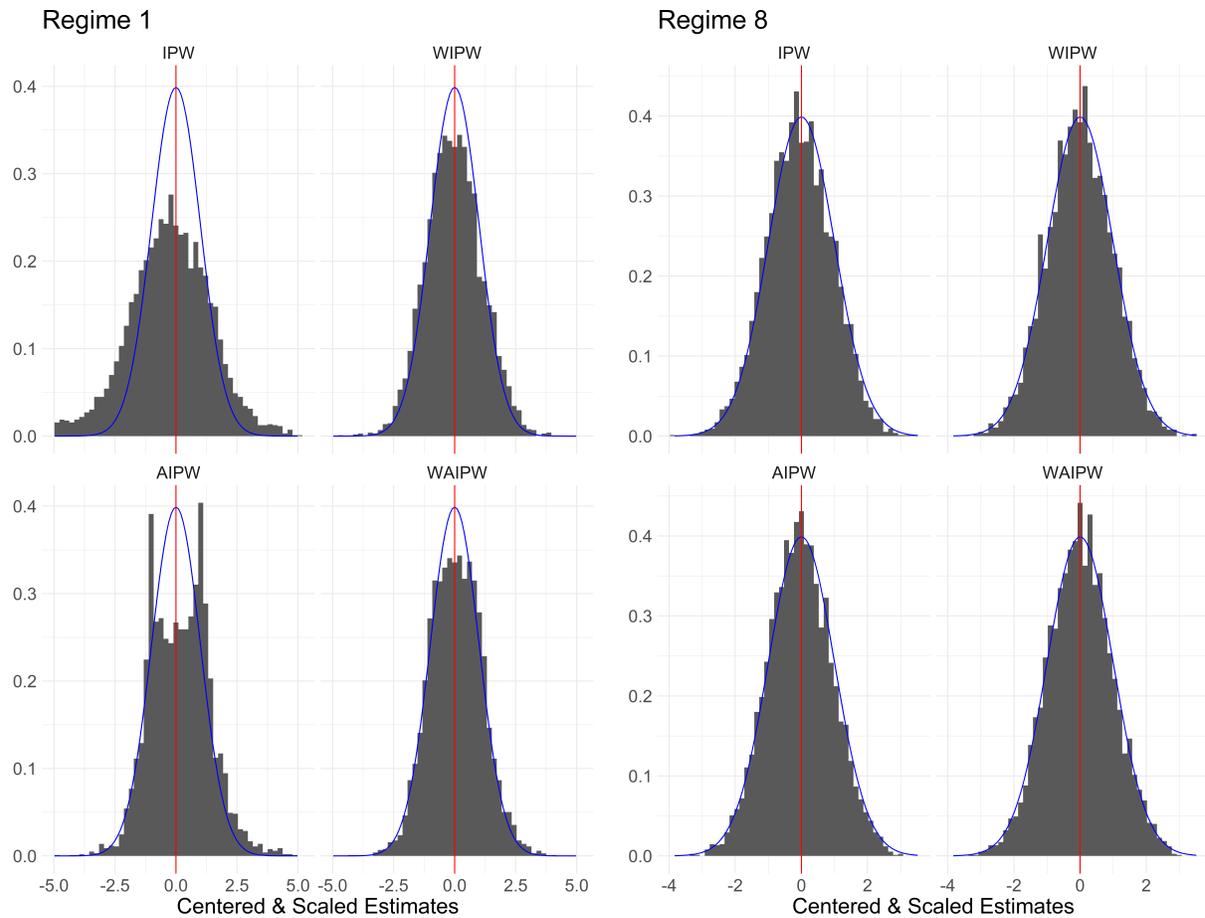}
  \caption{Monte Carlo distributions of centered and scaled estimates
    as in the theory of Section~\ref{s:theory} for
    selected estimators in the simulation in
    Section~\ref{ss:qlearnsims}.  The histograms correspond to the
    indicated estimator for $\mathcal{V}(\bd^1)$ and
    $\mathcal{V}(\bd^8)$ under sequential randomization using TS with $c_t=1$ for
    all $t$.  The vertical line indicates mean zero, and the density
    of a standard normal distribution is superimposed.}
  \label{f:ffour} 
\end{figure}



\section{Discussion}
\label{s:discussion}

We have proposed methods for RAR in SMARTs using TS, where
randomization can be up-front to regimes embedded in the trial or
performed sequentially at each stage. Simulation studies demonstrate
the benefits over nonadaptive randomization: improved outcomes for
subjects in the trial, improved ability to identify an optimal regime,
and little or no effect on post-trial inference on embedded regimes.
Choice of damping constant can dramatically change the aggressiveness
of TS; thus, the specific features and goals of a SMART should be
considered when specifying this parameter. When randomization is
up-front, the choice of estimator on which to base TS has some effect
on performance; WAIPW and AIPW estimators are more aggressive than
WIPW and IPW estimators. SMART-AR methods also yield good in- and
post-trial performance; however, the tuning parameters are less
intuitive and effective. For any SMART, we recommend simulating
multiple different adaptive randomization methods to see which best
aligns with the trial goals.  

The weighted versions of the IPW and AIPW estimators are preferred
over the unweighted estimators for post-trial inference.  Normalized
weighted estimators have sampling distribution closer to standard
normal, and yield improved coverage of confidence intervals and bounds
and ability to identify optimal embedded regimes.  When baseline and
intermediate subject variables that are correlated with outcome are
available, we recommend the WAIPW estimator for post-trial inference
when using RAR of any type.

\backmatter



\bibliographystyle{biom} \bibliography{biblio}


\appendix

\setcounter{equation}{0}
\renewcommand{\theequation}{A.\arabic{equation}}

\section*{Appendix A: Interim AIPW Estimator}

A limitation of using the WIPW or WAIPW estimators as the basis for
updating randomization probabilities is that only the data on subjects
in $\Data_{t-1}$ who have completed the trial and have $Y$ observed
are used.  Thus, partial information from subjects still progressing
through the trial is ignored.  To make use of this information, we
adopt the Interim AIPW (IAIPW) estimator of \citet*{Cole2022}, which
is derived by viewing the partial information through the lens of
monotone coarsening \citep{tsiatis2006}.

We present the estimator in the case where $K=2$ and
$n_{t-1} = \sumiN \Gamma_{t-1,i}$ ; a general formulation is presented
in \citet{Cole2022}.  In this setting, the weighted IAIPW (WIAIPW)
estimator for $\theta^j = \calV(\bd^j)$ based on $\Data_{t-1}$ is
   \begin{align}
\widehat{\theta}_{t}^{WIAIPW,j} = n_{t-1}^{-1} \sumiN & \Gamma_{t-1,i} 
W_{\tau_i}^j \left[ 
\frac{\Delta_{t-1,i} C_{i}^j Y_i}{ \pi_{\tau_i,2}^j(\overline{\bX}_{2,i}) 
\pi_{\tau_i,1}^j (\bX_{1,i}) \hatnu_{\Delta,t-1}}  + 
\left\{ 1 - \frac{\overline{C}_{1,i}^j  I(\kappa_{t-1,i} = 2)}{ 
\pi_{\tau_i,1}^j
     (\bX_{1,i}) \hatnu_{2,t-1}}  \right\} L_1^{j}(\bX_{1,i}) \right.  
    \nonumber \\
   &+\left. \left\{ \frac{\overline{C}_{1,i}^j  
    I(\kappa_{t-1,i} = 2)}{ \pi_{\tau_i,1}^j (\bX_{1,i}) \hatnu_{2,t-1}} 
 - \frac{\Delta_{t-1,i} C_{i}^j }{ \pi_{\tau_i,2}^j(\Xbar_{2,i}) 
 \pi_{\tau_i,1}^j
      (\bX_{1,i}) \hatnu_{\Delta,t-1}} \right\} L^j_2(\Xbar_{2,i}) \right], \nonumber
\end{align}
where $\hatnu_{\Delta,t-1} = \sumiN \Delta_{t-1,i}/n_{t-1}$ and
$\hatnu_{2,t-1} = \sumiN \Gamma_{t-1,i} I(\kappa_{t-1,i}=2)/n_{t-1}$
are estimators for $P(\Delta_{t-1} = 1|\Gamma_{t-1}=1, \Data_{t-1})$
and $P(\kappa_{t-1}=2 | \Gamma_{t-1}=1, \Data_{t-1})$, $W_{\tau_i}^j$
is a weight, and $L^j_k(\xbar_k)$, $k=1,2$, are arbitrary functions of
$\xbar_k$.  The efficient choice for the functions
$L^j_k(\xbar_k),\,k=1,2$ are derived from the $Q$-functions which can
be modeled and estimated based on a backward iterative scheme similar
to that presented for the WAIPW estimator, with modifications to the
definition of the pseudo outcomes for subjects in $\Data_{t-1}$ who
have not yet completed the trial; see \citet{Cole2022}.  As for (1)
and (2), the sampling distribution of
$\hattheta^{WIAIPW}_t = (\widehat{\theta}_{t}^{WIAIPW,1},\ldots,
\widehat{\theta}_{t}^{WIAIPW,m})^T$ can be approximated using large
sample theory analogous to that presented in the following two
sections. The general form of the weights is complex and in
preliminary simulations, weighted versions of the estimators did not
outperform a version where the weights are one for all $\tau$. Due to
this, the simulations in Section~\ref{s:sims} only use the
unweighted IAIPW as a basis for randomization for
simplicity. Moreover, we did not find the added precision from using
partial information to have a notable impact in randomization.

\setcounter{equation}{0}
\renewcommand{\theequation}{B.\arabic{equation}}

\section*{Appendix B: Theory on Weighted Estimators}


\par\noindent\textbf{\em Overview}

\newcommand{\bZ}{\bm{Z}}
\newcommand{\bTheta}{\bm{\Theta}}
\newcommand{\rp}{\mathbb{R}^p}
\newcommand{\R}{\mathbb{R}}

We use the theory for Martingale estimating functions
\citep[][MEFs]{godambe_1985,Heyde_1997,Godambe2010QuasilikelihoodAO}
to characterize the asymptotic behavior or our post-trial estimator
under Thompson Sampling.  For completeness, we state and prove the
results we will use.

Suppose that we observe the first $T$ elements of a discrete-time
stochastic process, $\bZ_t,\ldots, \bZ_T$, taking values in
$\mathcal{Z}$.  In addition, suppose that the law of this stochastic
process is $P$ and our interest is in estimating
$\btheta^* = \btheta(P) \in \bTheta \subseteq \mathbb{R}^p$.  We
assume there exists an $\rp$-valued process on $\bTheta$ given by
\begin{equation*}
\mathcal{M}_T(\btheta) = \sum_{t=1}^T \bM_t(\btheta),
\end{equation*}
such that $\left\lbrace \bM_t(\btheta^*)\right\rbrace_{t\ge 1}$ is a 
Martingale difference sequence with respect to the filtration $\left\lbrace \mathcal{F}_t
\right\rbrace_{t\ge 1}$, i.e., $E\left\lbrace 
\bM_t(\btheta^*) \big|\mathcal{F}_{t-1}
\right\rbrace = 0$ with probability one for all $t\ge 1$.  We construct
an estimator $\widehat{\btheta}_T$ as a root of 
$\mathcal{M}_T(\btheta) = 0$; more generally, it could be that no
exact root exists for finite $T$ and thus we select $\widehat{\btheta}_T$
as a minimizer of $||\mathcal{M}_T(\btheta)||_2$, where $||\cdot||_2$ 
denotes the Euclidean norm.  

\par\noindent\textbf{\em Consistency}

\newcommand{\calU}{\mathcal{M}_T}
\newcommand{\bthetaHatT}{\widehat{\btheta}_T}
\newcommand{\calUTheta}{\mathcal{M}_T(\btheta)}
\newcommand{\calUThetaStar}{\mathcal{M}_T(\btheta^*)}
\newcommand{\calUThetaBarT}{\mathcal{M}_T(\overline{\btheta}_T)}
\newcommand{\xiStarNegHalf}{\bxi_T^{-1/2}(\btheta^*)}
\newcommand{\xiStar}{\bxi_T(\btheta^*)}
\newcommand{\xiStarHalf}{\bxi_T^{1/2}(\btheta^*)}
\newcommand{\T}{^T}
\newcommand{\bxi}{\bm{\xi}}

We first consider conditions under which
$\widehat{\btheta}_T \inp \btheta^*$.  Define
$\bxi_T(\btheta) = \sum_{t=1}^T\mathrm{var}\left\lbrace
  \bM_t(\btheta)\big| \mathcal{F}_{t-1}\right\rbrace$ We will make use
of the following conditions.
\begin{itemize}
\item[(A0)] $\left\lbrace \bM_t(\btheta^*) \right\rbrace_{t\ge 1}$ is a 
Martingale difference sequence with respect to $\left\lbrace 
\mathcal{F}_t\right\rbrace_{t\ge 1}$.   
\item[(A1)] $\mathcal{M}_T(\btheta)$ is continuously differentiable
  for all $\btheta \in \bTheta$ and
  $\nabla_{\btheta} \mathcal{M}_T(\btheta)$ is invertible with
  probability one for all $\btheta\in\bTheta$ provided $T$ is
  sufficiently large.
\item[(A2)] The  minimum eigenvalue of $\bxi_T(\btheta^*)$ satisfies
$\lambda_{\min}\left\lbrace \bxi_T(\btheta^*) \right\rbrace 
\inp \infty,$ as $T\rightarrow \infty$. 
\item[(A3)] For any $\btheta \in \bTheta$ 
\begin{equation*}
  \sup_{\btheta\in\bTheta} \bigg|\bigg|
  \left\lbrace 
    \xiStarNegHalf \nabla_{\btheta}\mathcal{M}_t(\btheta) 
    \xiStarNegHalf 
  \right\rbrace^{-1}
  \bigg|\bigg| = O_P(1).   
\end{equation*}
\item[(A4)]  $\xiStarNegHalf \calUThetaStar = O_P(1)$.   
\end{itemize}

\begin{thm}
  Assume (A0)-(A4) and suppose that $\widehat{\btheta}_t$ satisfies
  $\mathcal{M}_t(\widehat{\btheta}_t) = o_P(1)$, then
  $\widehat{\btheta}_T \inp \btheta^*$ as $T\rightarrow \infty$.
\end{thm}
\begin{proof}
    Using a Taylor's expansion write
    \begin{equation*}
    \calU(\bthetaHatT) = \calU(\btheta^*) + 
    \nabla_{\btheta}\calUThetaBarT \left(\bthetaHatT - \btheta^*\right),
    \end{equation*}
    where $\overline{\btheta}_T$ is an intermediate point between
    $\bthetaHatT$ and $\btheta^*$.  Using (A1), we can re-arrange the
    above expression to obtain 
    \begin{eqnarray*}
    \bthetaHatT &=& \btheta^*  + \left\lbrace
    \nabla_{\btheta} \calUThetaBarT
    \right\rbrace^{-1}\left\lbrace 
    \calU(\bthetaHatT) - \calUThetaStar
    \right\rbrace \\[3pt]
    &=& \btheta^* + \left\lbrace
    \xiStarHalf \xiStarNegHalf \nabla_{\btheta}\calUThetaBarT \xiStarNegHalf
    \xiStarHalf 
    \right\rbrace^{-1} \left\lbrace
        \calU(\bthetaHatT) - \calUThetaBarT
    \right\rbrace\\[3pt]
    &=& \btheta^* + \xiStarNegHalf \left\lbrace
    \xiStarNegHalf \nabla_{\btheta} \calUThetaBarT \xiStarNegHalf
    \right\rbrace^{-1} 
    \xiStarNegHalf \left\lbrace
        \calU(\bthetaHatT) - \calUThetaStar
    \right\rbrace.
    \end{eqnarray*}
    From the minimum eigenvalue condition $||\xiStarNegHalf||_F = o_P(1)$,
    where $||\cdot||_F$ is the Frobenius norm. Furthermore, from (A3) it 
    follows that 
    \begin{equation*}
    \big|\big|\left\lbrace
    \xiStarNegHalf \nabla_{\btheta} \calUThetaBarT \xiStarNegHalf
    \right\rbrace^{-1}\big|\big| = O_P(1).  
    \end{equation*}
    Finally, we see that $\xiStarNegHalf\calU(\bthetaHatT)$ is the product
    of $o_P(1)$ terms and thus is $o_p(1)$, and
    $\xiStarNegHalf \calUThetaStar = O_P(1)$ from (A4). Thus, we have shown
    $\bthetaHatT = \btheta^* + o_P(1)$.  
\end{proof}

\par\noindent\textbf{\em Asymptotic Normality}

\newcommand{\bsigma}{\bm{\sigma}}
\newcommand{\bSigma}{\bm{\Sigma}}

To establish normality we replace (A4) with 
\begin{itemize}
\item[(A5)] Let $\bI_p$ denote the $p\times p$ identity matrix,
  $\xiStarNegHalf \mathcal{M}_T(\btheta^*) \inD \mathcal{N}(0,
  \bI_p)$.
\end{itemize}
In addition, we assume the following conditions.
\begin{itemize}
    \item[(A6)] For any consistent estimator 
    $\widetilde{\btheta}_T \inp \btheta^*$  
    \begin{equation*}
    \bigg|\bigg| 
    \xiStarNegHalf \left\lbrace
        \nabla_{\btheta}\calU(\widetilde{\btheta}_T) - 
        \nabla_{\btheta} \calUThetaStar 
    \right\rbrace
    \xiStarNegHalf
    \bigg|\bigg| \inp 0,
    \end{equation*}
    as $T\rightarrow \infty$.  
    \item[(A7)] There exists fixed and positive definite matrix
    $\bSigma$ such that 
    \begin{equation*}
     -\xiStarNegHalf \nabla_{\btheta} \calUThetaStar \xiStarNegHalf 
     \inp \bSigma,
    \end{equation*}
    as $T\rightarrow \infty$.
    \item[(A8)]  $\bthetaHatT$ satisfies $\xiStarNegHalf \calU(\bthetaHatT)
    \inp 0$ as $T\rightarrow \infty$.   
\end{itemize}

\begin{thm}
    Assume (A0)-(A3) and(A5)-(A8), then
    $\bSigma \xiStarHalf\left(\bthetaHatT - \btheta^*
    \right)\inD \mathcal{N}(0, \bI_p)$ as
    $T\rightarrow \infty$.  
\end{thm}
\begin{proof}
    As in our proof of consistency, we begin with a Taylor series expansion
    \begin{equation*}
    \calU(\bthetaHatT) = \calUThetaStar + \nabla_{\btheta}\calUThetaBarT
    \left(
        \bthetaHatT - \btheta^* 
    \right),
    \end{equation*}
    where $\overline{\btheta}_T$ is an intermediate point between
    $\bthetaHatT$ and $\btheta^*$.  re-arranging terms yields 
    \begin{eqnarray*}
        (\bthetaHatT - \btheta^*) &=& 
        \left\lbrace
            \nabla_{\btheta} \calUThetaBarT
        \right\rbrace^{-1} \left\lbrace
            \calU(\bthetaHatT) - \calUThetaStar 
        \right\rbrace  \\[5pt]
        \Leftrightarrow \xiStarHalf (\bthetaHatT - \btheta^*)
        &=& \xiStarHalf \left\lbrace
            \xiStarHalf \xiStarNegHalf 
            \nabla_{\btheta} \calUThetaBarT
            \xiStarNegHalf \xiStarHalf
        \right\rbrace^{-1} \\
      &\times& \left\lbrace
            \calU(\bthetaHatT) - \calUThetaStar 
        \right\rbrace\\[5pt]
        &=& \left\lbrace
            \xiStarNegHalf \nabla_{\btheta} \calUThetaBarT
            \xiStarNegHalf 
        \right\rbrace
        \xiStarNegHalf \left\lbrace
            \calU(\bthetaHatT) - \calUThetaStar
        \right\rbrace \\[5pt]
        &=& \left\lbrace
        -\xiStarNegHalf \nabla_{\btheta} \calUThetaBarT \xiStarNegHalf
        \right\rbrace^{-1} \xiStarNegHalf \calUThetaStar \\[3pt]
        && \quad + \left\lbrace
            \xiStarNegHalf \nabla_{\btheta} \calUThetaBarT \xiStarNegHalf
        \right\rbrace^{-1} \xiStarNegHalf \calU(\bthetaHatT) \\[5pt]
        &=& \bZ_T +\bE_T. 
    \end{eqnarray*}
    We show that $\bZ_T\inD \mathcal{N}(0,\bSigma^{-1}\bSigma^{-T})$
    and that $\bE_T = o_P(1)$, where $\bSigma^{-T}$ is the transpose of $\bSigma^{-1}$.  To obtain the limit for $\bZ_T$  we write
    \begin{align*}
    \bZ_T &= \left\lbrace
        -\xiStarNegHalf \nabla_{\btheta} \calUThetaBarT \xiStarNegHalf
    \right\rbrace^{-1} \xiStarNegHalf\calUThetaStar \\
    &= \left[
    -\xiStarNegHalf \nabla_{\btheta} \calUThetaStar \xiStarNegHalf
    + \xiStarNegHalf \left\lbrace
    \nabla_{\btheta} \calUThetaStar - \nabla_{\btheta} \calUThetaBarT
    \right\rbrace \xiStarNegHalf 
    \right]^{-1}   \\
 &\times\xiStarNegHalf\calUThetaStar,
    \end{align*}
    applying (A6), it follows that the term inside the square brackets 
    is $\bSigma^{-1} + o_P(1)$ so that
    \begin{equation*}
    \bZ_T = \bSigma^{-1}\xiStarHalf \calUThetaStar + o_P(1) 
    \inD \mathcal{N}(0, \bSigma^{-1}\bSigma^{-T}). 
    \end{equation*}
    It remains to show that $\bE_T = o_P(1)$.  From (A3) 
    $\left\lbrace \xiStarNegHalf \nabla_{\btheta}\calUThetaBarT \xiStarNegHalf \right\rbrace^{-1} = O_P(1)$.  Furthermore,
    $\xiStarNegHalf\calU(\widehat{\btheta}_{T}) = o_P(1)$ from (A8).
    This proves the result.  
\end{proof}

\subsection*{\textbf{\em Choosing the Weights}}

\par\noindent\textbf{\em WIPW Estimator After Up-Front Randomization}

Consider the WIPW estimator for the value $\theta^j$ of $j$th embedded
regime in  (\ref{eq:IPWrand}) at the end of the trial under the conditions of
Section~\ref{s:theory},
so that subjects are indexed by $t$.  We require
that the weights $W^j_t$ be chosen as functions of $\Data_{t-1}$ so
that (i) the estimating equations are conditionally unbiased,
$E\{ M_t^j (\Xbar_{K,t},\Abar_{K,t},Y_t; \theta^j) | \Data_{t-1}\} =
0$, and (ii) the variance is stabilized,
$E[ \{M_t^j (\Xbar_{K,t},\Abar_{K,t},Y_t; \theta^j)\}^2 |
\Data_{t-1}]$ $= \sigma > 0$ for all $t$, where
$M_t^j (\Xbar_{K,t},\Abar_{K,t},Y_t; \theta^j)$ is the estimating
function  (\ref{eq:MestIPW}) given by
$$M_t^j (\Xbar_{K,t},\Abar_{K,t},Y_t; \theta^j) = \frac{ W^j_t
  C^j_t}{ \{ \prod_{k=2}^K \pi_{t,k}^j(\Xbar_{k,t})\}
  \pi_{t,1}^j(\bX_{1,t}) } (Y_t - \theta^j).$$
As in Section~\ref{s:theory},
take $\Xbar_k^*(\bd^j) = \{\bX_1, \bX_2^*(\bd^j), \ldots, \bX_k^*(\bd^j)\}$, $k=1,\ldots,K$. 

To show (i),
when $C^j_t=1$, the consistency assumption given in Section~\ref{ss:SMART}
implies that $Y_t = Y^*(\bd^j)$ and 
$\Xbar_{K,t} = \Xbar_K^*(\bd^j)$.  Thus, if $W^j_t$ is a
function of $\Data_{t-1}$, 
\begin{align*}
  E&\{ M_t^j (\Xbar_{K,t},\Abar_{K,t},Y_t; \theta^j)  | \Data_{t-1}\}\\
   &=  W^j_t E \left[ \left. \frac{ C^j_t \{Y^*(\bd^j) - \theta^j\}}
     {\prod_{k=2}^K\pi_{t,k}^j\{\Xbar_k^*(\bd^j)\}  \pi_{t,1}^j(\bX_1)} \right|   \Data_{t-1}\right] \\
   &= W^j_t E \left( \left. E\left[ \left. \frac{ C^j_t  \{Y^*(\bd^j)
     - \theta^j\} }
     { \prod_{k=2}^K\pi_{t,k}^j\{\Xbar_k^*(\bd^j)\}  \pi_{t,1}^j(\bX_1)}\right| \mathcal{W}^*,
                             \Data_{t-1} \right] \right| \Data_{t-1}\right) \\
    &= W^j_t  E \left[ \left. \frac{ E(C^j_t| \mathcal{W}^*,
      \Data_{t-1})}
      {  \prod_{k=2}^K\pi_{t,k}^j\{\Xbar_k^*(\bd^j)\}  \pi_{t,1}^j(\bX_1) } \{Y^*(\bd^j) - \theta^j\} \right|
      \Data_{t-1}\right]\\*[0.1in]
   &= W^j_t  E \left[ \left. \frac{\prod_{k=2}^K\pi_{t,k}^j\{\Xbar_k^*(\bd^j)\}  \pi_{t,1}^j(\bX_1)}
     {  \prod_{k=2}^K\pi_{t,k}^j\{\Xbar_k^*(\bd^j)\}  \pi_{t,1}^j(\bX_1) } \{Y^*(\bd^j) - \theta^j\} \right|
      \Data_{t-1}\right]\\
      &= W^j_t E\{ Y^*(\bd^j) - \theta^j \} = 0,
  \end{align*}
  where we have used the result analogous to that in Section~6.4.3 of
  \citet{tsiatis2020dynamic} that, under the consistency, positivity,
  and sequential ignorability assumptions, 
  $E(C^j_t| \mathcal{W}^*, \Data_{t-1}) =  \prod_{k=2}^K
\pi_{t,k}^j\{\Xbar_k^*(\bd^j)\}  \pi_{t,1}^j(\bX_1)$ and the fact that $\mathcal{W}^*$ and thus $Y^*(\bd^j)$
  is independent of $\Data_{t-1}$. 

  To find $W^j_t$ satisfying (ii), we have, using the consistency assumption,
  $(C^j_t)^2 = C^j_t$, and that $E(C^j_t| \mathcal{W}^*, \Data_{t-1}) =  \prod_{k=2}^K
\pi_{t,k}^j\{\Xbar_k^*(\bd^j)\}  \pi_{t,1}^j(\bX_1)$, that
  \begin{align}
E&\{ M_t^j (\Xbar_{K,t},\Abar_{K,t},Y_t; \theta^j)^2  |  \Data_{t-1}\}
   =  (W^j_{t} )^2 E \left( \left. \frac{ C^j_t \{Y^*(\bd^j) -\theta^j\}^2 }
   { \big[\prod_{k=2}^K\pi_{t,k}^j\{\Xbar_k^*(\bd^j)\}  \pi_{t,1}^j(\bX_1)\big]^2 } \right|      \Data_{t-1}\right) \nonumber\\
 &= (W^j_t)^2  E \left( \left. \frac{ E(C^j_t| \mathcal{W}^*,
   \Data_{t-1})}
   { \big[ \prod_{k=2}^K\pi_{t,k}^j\{\Xbar_k^*(\bd^j)\}  \pi_{t,1}^j(\bX_1)\big]^2 } \{Y^*(\bd^j) - \theta^j\}^2 \right|   \Data_{t-1}\right) \nonumber\\
    &= (W^j_t)^2  E \left[ \left. \frac{\{Y^*(\bd^j) - \theta^j\}^2}
      { \prod_{k=2}^K\pi_{t,k}^j\{\Xbar_k^*(\bd^j)\}  \pi_{t,1}^j(\bX_1)  }\right|
       \Data_{t-1}\right] \label{eq:varrequire}
\end{align}    
as in \ref{eq:Mvar}.

To ensure (ii), we thus want to choose $W^j_t$ so that
(\ref{eq:varrequire}) is a constant that can be $j$-specific but that
does not depend on $\Data_{t-1}$.  For definiteness and simplicity,
consider $K=2$; the following argument extends to general $K$.  As in
Section~\ref{s:theory}
for the cancer pain management SMART, randomization
at stage 1 ordinarily does not depend on covariate information, which
implies that $\pi_{t,1}^j(\bX_1) = \pi^j_{t,1}$ depending on
$\Data_{t-1}$ and $j$ but not $X_1$.  Randomization at stages beyond
stage 1 is typically to sets of feasible treatment options within each
level of a  binary variable such as
response status, and depends on no other covariate information.  Let $R^*(\bd^j)$ be the potential binary response status
variable taking on values 0 and 1 for a subject whose first stage
treatment assignment is consistent with regime $\bd^j$, where
$R^*(\bd^j)$ is a function of or component of $\Xbar_2^*(\bd^j)$
($R^*(\bd^j) = X_{22}^*(\bd^j) $ in the cancer pain management SMART).
By the consistency assumption, the observed response status $R$ is equal
to $R^*(\bd^j)$ when $C^j_t=1$.  Then, as in the cancer pain management
SMART, it follows that
$\pi_{t,2}^j\{\Xbar_2^*(\bd^j)\} = \pi^{j(s)}_{t,2}$ when
$R^*(\bd^j)=s$, $s=0,1$, where for each $r$, $\pi^{j(s)}_{t,2}$ depends
on $\Data_{t-1}$ and $j$ but not on $\Xbar_2^*(\bd^j)$, so that
\begin{equation}
\pi_{t,2}^j\{\Xbar_2^*(\bd^j)\} = I\{ R^*(\bd^j)=0\} \pi^{j(0)}_{t,2}
+  I\{ R^*(\bd^j)=1\}\pi^{j(1)}_{t,2}.
\label{eq:piXbar}
\end{equation}
Because $\mathcal{W}^* \independent\Data_{t-1}$, generalizing
(\ref{eq:regime1var}), it follows that (\ref{eq:varrequire}) is equal to
\begin{equation}
(W^j_{t} )^2 \left( \frac{ \mu^{j(1)}} {\pi^{j(1)}_{t,2}\pi_{t, 1}^j} + \frac{ \mu^{j(0)}} {\pi^{j
      (0)}_{t,2}\pi_{t, 1}^j} \right), 
\label{eq:constantvar}
  \end{equation}
where $\mu^{j(s)} = E[I\{R^*(\bd^j) = s\} \{ Y^*(\bd^j) - \theta^j\}^2]$, $s = 0, 1$.
Thus, we wish to choose $W^j_{t}$ so that (\ref{eq:constantvar}) is
equal to a constant depending only on $j$ and not on $\Data_{t-1}$.
Denoting this constant as $\Xi^j$, the weights should be chosen so
that
\begin{equation}
W^j_{t} = (\Xi^j)^{1/2} \left( \frac{ \mu^{j(1)}} {\pi^{j(1)}_{t,2}\pi_{t, 1}^j} + \frac{ \mu^{j(0)}} {\pi^{j
      (0)}_{t,2}\pi_{t, 1}^j} \right)^{-1/2}.  
\label{eq:weightsj}
  \end{equation}

  Of course, in practice the constant $\Xi^j$ is unknown.  If there is
  a burn-in period of length $t^*$ weeks during which
  nonadaptive randomization takes place with $W^j_t \equiv 1$,
  $t \leq t^*$, from (\ref{eq:constantvar}), it must be the case that $\Xi^j$ satisfies
  \begin{equation}
\Xi^j = \frac{ \mu^{j(1)}}{\pi^{j(1)}_{t^*,2}\pi_{t^*, 1}^j} + \frac{ \mu^{j(0)}} {\pi^{j(0)}_{t^*,2}\pi_{t^*, 1}^j},
\label{eq:Xij}
\end{equation}
where the randomization probabilities during the burn-in period are fixed, so
that $\pi^{j(s)}_{t^*,2}$, $s=0,1,$ and $\pi_{t^*, 1}^j$ are known constants not
depending on the data.  Thus, if $\mu^{j(s)}$, $s=0, 1$, were known,
so that $\Xi^j$ is known, the weights for $t \geq t^* + 1$ should be
chosen as in (\ref{eq:weightsj}) with $\Xi^j$ as in (\ref{eq:Xij}).

These observations suggest our proposed scheme for implementation in
practice.  We propose to obtain weights in this spirit by estimating
$\Xi^j$ by an estimator $\widehat{\Xi}_{t^*}^j$, say, based on the
data from a burn-in period of length $t^*$ weeks during which
nonadaptive randomization takes place with $W^j_t \equiv 1$,
$t \leq t^*$, and then treat the estimate as fixed and known
henceforth.  Specifically, let $\hatmu^{j(s)}_{t}$ be estimators for
$\mu^{j(s)}$, $s=0, 1$, at any time $t$ based on $\Data_{t-1}$, as in
the case of the cancer pain management SMART (in the original
notation) in Section~\ref{s:theory}; with the indexing scheme here (by $t$),
$\hatmu^{j(s)}_{t} = t^{-1} \sum_{u=1}^t  I(R_u = s) C^j_u
(Y_u-\widetilde{\theta}^j_{t})^2/( \pi^{j(s)}_{u,2} \pi^j_{u,1})$,
$s=0,1$, where $\widetilde{\theta}^j_{t}$ is an estimator for
$\theta^j$ using $\Data_{t-1}$ (e.g., the unweighted IPW estimator).  Define
  $$\widehat{\Xi}^j_{t} = \frac{ \hatmu^{j(1)}_{t}} {\pi^{j(1)}_{t,2}\pi_{t, 1}^j} + \frac{ \hatmu^{j(0)}_{t}} {\pi^{j
      (0)}_{t,2}\pi_{t, 1}^j}.$$
  Then (\ref{eq:Xij}) implies that the estimator
  $\widehat{\Xi}_{t^*,j}$ for $\Xi^j$ based on the burn-in data is
  $$\widehat{\Xi}^j_{t^*} =
\frac{ \hatmu^{j(1)}_{t^*}} {\pi^{j(1)}_{t^*,2}\pi_{t^*, 1}^j} + \frac{ \hatmu^{j(0)}_{t^*}} {\pi^{j
      (0)}_{t^*,2}\pi_{t^*, 1}^j},$$
 where as above $\pi^{j(s)}_{t^*,2}$, $s=0, 1$, and $\pi_{t^*, 1}^j$ are known constants not
  depending on the data.   Treating $\widehat{\Xi}^j_{t^*}$ as fixed and known, for $t \geq
  t^*+1$, from (\ref{eq:weightsj}), take $W^j_t =
  (\widehat{\Xi}^j_{t^*}/\widehat{\Xi}^j_{t})^{1/2}$.  Because
  $\widehat{\Xi}^j_{t^*}$ depends only on $\Data_{t^*} \subseteq
  \Data_{t-1}$ for $t \geq t^*+1$, and $\widehat{\Xi}^j_{t}$ depends only
    on $\Data_{t-1}$, the weights  $W^j_t$ for $t\geq t^*+1$ depend
    on $\Data_{t-1}$, as required.  

\par\noindent\textbf{\em WIPW Estimator After Sequential Randomization}  
 
Again consider $K=2$ for definiteness.  The preceding scheme for
obtaining weights assumes that up-front randomization is used, so
that, returning to the original notation, for a subject who entered
the SMART at time $\tau$, the probabilities in the denominator of
(1) 
for that subject's contribution 
are determined at that point based on $\Data_{\tau-1}$.  Thus, if
up-front randomization was used, the estimators for $\mu^{(s)}_j$,
$s=0,1$, and the corresponding weight $W^j_\tau$ for a such a subject
depend on $\pi^{j(s)}_{\tau,2}$, $s=0,1$, and $\pi_{\tau, 1}^j$.

If the sequential randomization scheme is used, this calculation
should be modified.  In obtaining the WIPW estimator for $\theta^j$ at
the end of the trial, for a subject who entered the SMART at time
$\tau$ and then reached stage 2 at time $\tau+v$, say, the
probabilities in the denominator of (1) 
should be $\pi_{\tau, 1}^j$ and $\pi^{j(s)}_{\tau+v,2}$, $s=0,1$.
Likewise, the estimators for $\mu^{(s)}_j$, $s=0,1$, and the
corresponding weight should be calculated based on
$\pi^{j(s)}_{\tau+v,2}$, $s=0,1$, and $\pi_{\tau, 1}^j$, so that the
weight for this subject depends on $\Data_{\tau+v-1}$ and thus can be
written as $W^j_{\tau+v}$.  When $\tau+v < t^*$,
$W^j_{\tau+v} = 1$.

\par\noindent\textbf{\em WAIPW Estimator After Up-Front Randomization}

Analogous to the developments for the WIPW estimator for $\theta^j$
for embedded regime $j$, using the indexing scheme (by $t$) in Section~\ref{s:theory},
we identify the corresponding estimating function
$M_t^j (\Xbar_{K,t},\Abar_{K,t},Y_t; \theta^j)$, demonstrate that
(i) $E\{ M_t^j (\Xbar_{K,t},\Abar_{K,t},Y_t; \theta^j) | \Data_{t-1}\} =
0$, and determine weights $W^{A,j}_t$ so that (ii) the variance
$E[ \{M_t^j (\Xbar_{K,t},\Abar_{K,t},Y_t; \theta^j)\}^2 |
\Data_{t-1}]$ is a constant depending only on $j$. We demonstrate for
$K=2$; the argument extends to general $K$.

From \ref{eq:AIPWrand}, for $K=2$ it is straightforward that
\begin{align*}
  M_t^j& (\Xbar_{2,t},\Abar_{2,t},Y_t; \theta^j) 
  = W^{A,j}_t \left[ \frac{C^j_t
           Y_t}{\pi^j_{t,2}(\Xbar_{2,t})\pi^j_{t,1}(\bX_{1,t})} -\left\{
           \frac{\overline{C}^j_{1,t}}{\pi^j_{t,1}(\bX_{1,t})} -1 \right\}    \hatQ^j_{1,t}(\bX_{1,t}) \right.\\
   &-\left. \left\{ \frac{C^j}{\pi^j_{t,2}(\Xbar_{2,t})\pi^j_{t,1}(\bX_{1,t}1)} -\frac{\overline{C}^j_{1,t}}{\pi^j_{t,1}(\bX_{1,t})} \right\}
     \hatQ^j_{2,t}(\Xbar_{2,t}) - \theta^j\right],
  \end{align*}
  where $\overline{C}^j_{1,t}  = I\{A_{1,t} =
  d^j_1(\bX_{1,t})\}$, and for brevity we write
  $\hatQ^j_{2,t}(\Xbar_{2,t}) = Q_2\{ \Xbar_{2,t}, \dbar_2^j(\Xbar_{2,t};
  \hatbeta_{2,t}\}$ and
  $\hatQ^j_{1,t}(\bX_{1,t}) = Q^j_1\{\bX_{1,t}, d^j_1(\bX_{1,t});
  \hatbeta^j_{1,t}\}$, which are fitted using $\Data_{t-1}$.  To show
  (i), using the consistency assumption and rearranging, 
  similar to the argument for the WIPW estimator, if  $W^{A,j}_t$ is a
function of $\Data_{t-1}$, 
  \begin{align}
  E&\{  M_t^j (\Xbar_{2,t},\Abar_{2,t},Y_t; \theta^j) |\Data_{t-1}\}
     =  W^{A,j}_t \left[ \vphantom{\frac{C^j_t}{\pi_{t,2}^j\{\Xbar_2^*(\bd^j)\}\pi_{t,1}^j(\bX_1)}}
    E\big[ \{ Y^*(\bd^j) - \theta^j\} | \Data_{t-1}\big] \right. \label{eq:term1} \\
  & + E\left(\left. \left[ \frac{\overline{C}^j_{1,t}}{\pi^j_{t,1}(\bX_1)} -1\right]
    \{Y^*(\bd^j) - \hatQ^j_{1,t}(\bX_1) \} \right| \Data_{t-1}\right)\label{eq:term2}\\
     &+ \left.\left(\left.\left[ \frac{C^j_t}{\pi_{t,2}^j\{\Xbar_2^*(\bd^j)\}\pi_{t,1}^j(\bX_1)} -
       \frac{\overline{C}^j_{1,t}}{\pi^j_{t,1}(\bX_1)}\right]
       [Y^*(\bd^j) - \hatQ^j_{2,t}\{\Xbar_k^*(\bd^j)\} ]\right| \Data_{t-1}\right) \right].\label{eq:term3}
  \end{align}
  Using $\mathcal{W}^* \independent\Data_{t-1}$, it follows that
  $E[ \{ Y^*(\bd^j) - \theta^j\} | \Data_{t-1}] = E[ \{ Y^*(\bd^j) -
  \theta^j\} ] = 0$, so that the conditional expectation in (\ref{eq:term1}) is
  zero.  Using arguments similar to those for the WIPW estimator, the
  conditional expectations in (\ref{eq:term2}) and (\ref{eq:term3}) can also be
  shown to be equal to zero using
  $E(C^j_t | \mathcal{W}^* ,\Data_{t-1}) =
  \pi_{t,2}^j\{\Xbar_2^*(\bd^j)\}\pi_{t,1}^j(\bX_1)$ and
  $E(\overline{C}^j_{1,t} | \mathcal{W}^*,\Data_{t-1}) =
\pi_{t,1}^j(\bX_1)$.  Thus, (i) holds.

To find $E[ \{M_t^j (\Xbar_{K,t},\Abar_{K,t},Y_t; \theta^j)\}^2 |
\Data_{t-1}]$, first note that the conditional (on
$\Data_{t-1}$) expectations of crossproduct terms in $\{M_t^j
(\Xbar_{K,t},\Abar_{K,t},Y_t; \theta^j)\}^2$ are equal to zero by
similar arguments.  For example, it is straightforward that
\begin{align*}
 (W^{A,j}_t)^2 E\left( \left.\left[ \frac{\overline{C}^j_{1,t}}{\pi^j_{t,1}(\bX_1)} -1\right]
 \{ Y^*(\bd^j) - \theta^j\}   \{Y^*(\bd^j) - \hatQ^j_{1,t}(\bX_1) \} \right| \Data_{t-1}\right) = 0
  \end{align*}
using $\mathcal{W}^* \independent\Data_{t-1}$ and  $E(\overline{C}^j_{1,t} | \mathcal{W}^*,\Data_{t-1}) =
\pi_{t,1}^j(\bX_1)$.  Thus,
\begin{align}
E&[ \{M_t^j (\Xbar_{K,t},\Abar_{K,t},Y_t; \theta^j)\}^2 |
\Data_{t-1}] = (W^{A,j}_t)^2 \left[ \vphantom{\frac{C^j_t}{\pi_{t,2}^j\{\Xbar_2^*(\bd^j)\}\pi_{t,1}^j(\bX_1)}}
    E[ \{ Y^*(\bd^j) - \theta^j\}^2 | \Data_{t-1}] \right. \label{eq:term1var} \\
  & + E\left(\left. \left[ \frac{\overline{C}^j_{1,t}}{\pi^j_{t,1}(\bX_1)} -1\right]^2
    \{Y^*(\bd^j) - \hatQ^j_{1,t}(\bX_1) \}^2 \right| \Data_{t-1}\right)\label{eq:term2var}\\
     &+ \left.\left(\left.\left[ \frac{C^j_t}{\pi_{t,2}^j\{\Xbar_2^*(\bd^j)\}\pi_{t,1}^j(\bX_1)} -
       \frac{\overline{C}^j_{1,t}}{\pi^j_{t,1}(\bX_1)}\right]^2
       [Y^*(\bd^j) - \hatQ^j_{2,t}\{\Xbar_k^*(\bd^j)\} ]^2\right| \Data_{t-1}\right) \right].\label{eq:term3var}
  \end{align}
  In (\ref{eq:term1var}), using $\mathcal{W}^*
  \independent\Data_{t-1}$, $E[ \{ Y^*(\bd^j) - \theta^j\}^2 |
\Data_{t-1}] = E[ \{ Y^*(\bd^j) - \theta^j\}^2]$.  Again using $\mathcal{W}^* \independent\Data_{t-1}$
and the fact that $E(\overline{C}^j_{1,t} | \mathcal{W}^*,\Data_{t-1}) =
\pi_{t,1}^j(\bX_1)$ implies that 
$$E\left[\left\{ \left.\frac{\overline{C}^j_{1,t}}{\pi^j_{t,1}(\bX_1)}
    -1\right\}^2 \right| \mathcal{W}^*, \Data_{t-1}\right]  = \frac{1 -
  \pi^j_{t,1}(\bX_1)}{\pi^j_{t,1}(\bX_1)},$$
so that the conditional expectation in (\ref{eq:term2var}) is
$$E\left[ \left.\{Y^*(\bd^j) - \hatQ^j_{1,t}(\bX_1) \}^2 \left\{\frac{1 -
  \pi^j_{t,1}(\bX_1)}{\pi^j_{t,1}(\bX_1)} \right\}\right| \Data_{t-1}\right].$$
Finally, using $E(C^j_t | \mathcal{W}^* ,\Data_{t-1}) =
  \pi_{t,2}^j\{\Xbar_2^*(\bd^j)\}\pi_{t,1}^j(\bX_1)$ and
  $E(\overline{C}^j_{1,t} | \mathcal{W}^*,\Data_{t-1}) =
  \pi_{t,1}^j(\bX_1)$, it is straightforward to derive that
  $$E \left(\left.\left[ \frac{C^j_t}{\pi_{t,2}^j\{\Xbar_2^*(\bd^j)\}\pi_{t,1}^j(\bX_1)} -
       \frac{\overline{C}^j_{1,t}}{\pi^j_{t,1}(\bX_1)}\right]^2
   \right| \mathcal{W}^*,\Data_{t-1} \right) = \frac{
   1-\pi_{t,2}^j\{\Xbar_2^*(\bd^j)\} }{\pi_{t,2}^j\{\Xbar_2^*(\bd^j)\}\pi^j_{t,1}(\bX_1)}.$$
Using this result, the conditional expectation in (\ref{eq:term3var})
is
$$E\left( \left. [Y^*(\bd^j) - \hatQ^j_{2,t}\{\Xbar_k^*(\bd^j)\} ]^2 \frac{
   1-\pi_{t,2}^j\{\Xbar_2^*(\bd^j)\}
 }{\pi_{t,2}^j\{\Xbar_2^*(\bd^j)\}\pi^j_{t,1}(\bX_1)} \right| \Data_{t-1}\right).$$
Substituting the foregoing results in  (\ref{eq:term1var}),
(\ref{eq:term2var}), and  (\ref{eq:term3var}) yields
 \begin{align}
 E&[ \{M_t^j (\Xbar_{K,t},\Abar_{K,t},Y_t; \theta^j)\}^2 |
 \Data_{t-1}] = (W^{A,j}_t)^2 \left\{ \vphantom{\frac{C^j_t}{\pi_{t,2}^j\{\Xbar_2^*(\bd^j)\}\pi_{t,1}^j(\bX_1)}}
    E\big[ \{ Y^*(\bd^j) - \theta^j\}^2\big] \right. \nonumber \\
   &+ E\left[ \left.\{Y^*(\bd^j) - \hatQ^j_{1,t}(\bX_1) \}^2 \left\{\frac{1 -
   \pi^j_{t,1}(\bX_1)}{\pi^j_{t,1}(\bX_1)} \right\}\right|
     \Data_{t-1}\right] \label{eq:terms}\\
 &+\left.E\left( \left. [Y^*(\bd^j) - \hatQ^j_{2,t}\{\Xbar_k^*(\bd^j)\} ]^2 \left[\frac{
    1-\pi_{t,2}^j\{\Xbar_2^*(\bd^j)\}
  }{\pi_{t,2}^j\{\Xbar_2^*(\bd^j)\}\pi^j_{t,1}(\bX_1)} \right]\right|
   \Data_{t-1}\right)\right\}.  \nonumber 
   \end{align}

Define $\pi^j_{t,1}$ and $\pi^{j(s)}_{t,2}$, $s=0,1$, which depend
only on $\Data_{t-1}$ and $j$, and $R^*(\bd^j)$ as for the WIPW estimator, 
so that
$\pi_{t,2}^j\{\Xbar_2^*(\bd^j)\} = I\{ R^*(\bd^j)=0\} \pi^{j(0)}_{t,2}
+  I\{ R^*(\bd^j)=1\}\pi^{j(1)}_{t,2}$ as in (\ref{eq:piXbar}).  Then
(\ref{eq:terms}) can be rewritten as
\begin{align}
 (W^{A,j}_t)^2 &\left\{ \vphantom{\frac{C^j_t}{\pi_{t,2}^j\{\Xbar_2^*(\bd^j)\}\pi_{t,1}^j(\bX_1)}}
    E\big[ \{ Y^*(\bd^j) - \theta^j\}\big]^2 \right. \label{eq:term1varnew} \\
               &+ \left(\frac{1 -\pi^j_{t,1}}{\pi^j_{t,1}} \right)
                 E\left[ \left.\{Y^*(\bd^j) - \hatQ^j_{1,t}(\bX_1) \}^2 \right| \Data_{t-1}\right]   \label{eq:term2varnew}\\
               &+\left( \frac{1-\pi^{j(0)}_{t,2} }{\pi^{j(0)}_{t,2}\pi^j_{t,1}} \right)
                 E\left( \left. I\{R^*(\bd^j)=0\}  [Y^*(\bd^j) - \hatQ^j_{2,t}\{\Xbar_k^*(\bd^j)\} ]^2 \right|
                 \Data_{t-1}\right)   \label{eq:term30varnew} \\
            &+\left. \left( \frac{1-\pi^{j(1)}_{t,2} }{\pi^{j(1)}_{t,2}\pi^j_{t,1}} \right)
                 E\left( \left. I\{R^*(\bd^j)=1\}  [Y^*(\bd^j) - \hatQ^j_{2,t}\{\Xbar_k^*(\bd^j)\} ]^2 \right|
   \Data_{t-1}\right)   \right\}.  \label{eq:term31varnew} 
   \end{align}
Write the expectations in
(\ref{eq:term1varnew})--(\ref{eq:term31varnew}) respectively as  $\nu^j
   = E[ \{ Y^*(\bd^j) - \theta^j\}]^2$, $\nu^j_{1} = E[ \{Y^*(\bd^j) - \hatQ^j_{1,t}(\bX_1) \}^2 | \Data_{t-1}]$, and 
$$\nu^{j(s)}_{2} = E \left(\left. I\{R^*(\bd^j)=1\}  [Y^*(\bd^j) -
  \hatQ^j_{2,t}\{\Xbar_k^*(\bd^j)\} ]^2 \right| \Data_{t-1}\right), \,\,\,
s=0,1,$$
so that the variance (\ref{eq:terms}) can be written succinctly as 
\begin{align}
 (W^{A,j}_t)^2 \left\{ \nu^j + \nu^j_{1} \left(\frac{1
  -\pi^j_{t,1}}{\pi^j_{t,1}} \right) + \nu^{j(0)}_{2} \left(
  \frac{1-\pi^{j(0)}_{t,2} }{\pi^{j(0)}_{t,2}\pi^j_{t,1}} \right)
  + \nu^{j(1)}_{2} \left(
  \frac{1-\pi^{j(1)}_{t,2} }{\pi^{j(1)}_{t,2}\pi^j_{t,1}} \right) \right\}.
  \label{eq:waipwvar}
\end{align}
Thus, as for the WIPW estimator, we wish to choose $W^{A,j}_t$  so
that (\ref{eq:waipwvar}) is a constant depending on $j$ but not
$\Data_{t-1}$.  Denoting this constant as $\Xi^{A,j}$, the weights
should be chosen so that
 \begin{equation}
W^{A,j}_t = (\Xi^{A,j})^{1/2} \left\{ \nu_j + \nu^j_{j} \left(\frac{1
  -\pi^j_{t,1}}{\pi^j_{t,1}} \right) + \nu^{j(0)}_{2} \left(
  \frac{1-\pi^{j(0)}_{t,2} }{\pi^{j(0)}_{t,2}\pi^j_{t,1}} \right)
  + \nu^{j(1)}_{2} \left(
  \frac{1-\pi^{j(1)}_{t,2} }{\pi^{j(1)}_{t,2}\pi^j_{t,1}} \right)
\right\}^{-1/2}.
\label{eq:waipwweight}
   \end{equation}
With a burn-in period of length $t^*$ weeks of nonadaptive
randomization with $W^{A,j}_t \equiv 1$, from (\ref{eq:waipwweight}),
it must be that
\begin{equation}
\Xi^{A,j} = \nu^j + \nu^j_{1} \left(\frac{1
  -\pi^j_{t^*,1}}{\pi^j_{t^*,1}} \right) + \nu^{j(0)}_{2} \left(
  \frac{1-\pi^{j(0)}_{t^*,2} }{\pi^{j(0)}_{t^*,2}\pi^j_{t^*,1}} \right)
  + \nu^{j(1)}_{2} \left(
  \frac{1-\pi^{j(1)}_{t^*,2} }{\pi^{j(1)}_{t,2}\pi^j_{t^*,1}} \right),
\label{eq:XiAij}
\end{equation}
where as before the randomization probabilities during the burn-in period are fixed, so
that $\pi^{j(s)}_{t^*,2}$, $s=0,1,$ and $\pi_{t^*, 1}^j$ are known constants not
depending on the data.  Thus, if $\nu^j$, $\nu^j_{1}$, and
$\nu^{j(s)}_{2}$, $s=0, 1$, were known,
so that $\Xi^{A,j}$ is known, the weights for $t \geq t^* + 1$ should be
chosen as in (\ref{eq:waipwweight}) with $\Xi^{A,j}$ as in (\ref{eq:XiAij}).  

Analogous to the approach for the WIPW estimator, the proposed
implementation scheme is as follows.  Obtain weights by estimating
$\Xi^{A,j}$ by an estimator $\widehat{\Xi}^{A,j}_{t^*}$ based on the data
from a burn-in period of length $t^*$ weeks with nonadaptive
randomization with $W^{A,j}_t \equiv 1$, $t \leq t^*$, and then going
forward treat the estimate as fixed and known.  Letting 
$\widehat{\nu}^j_{t}$, $\widehat{\nu}^j_{t,1}$, and
$\widehat{\nu}^{j(s)}_{t,2}$, $s=0, 1$, be estimators for 
$\nu^j$, $\nu^j_{1}$, and
$\nu^{j(s)}_{2}$, $s=0, 1$, at any time $t$ based on $\Data_{t-1}$,
define
\begin{equation}
\widehat{\Xi}^{A,j}_{t} = \widehat{\nu}^j_{t} + \widehat{\nu}^j_{t,1} \left(\frac{1
  -\pi^j_{t,1}}{\pi^j_{t,1}} \right) + \widehat{\nu}^{j(0)}_{t,2} \left(
  \frac{1-\pi^{j(0)}_{t,2} }{\pi^{j(0)}_{t^*,2}\pi^j_{t,1}} \right)
  + \widehat{\nu}^{j(1)}_{t,2} \left(
    \frac{1-\pi^{j(1)}_{t,2} }{\pi^{j(1)}_{t,2}\pi^j_{t,1}} \right).
\label{eq:hatXiAij}
  \end{equation}
Then (\ref{eq:XiAij}) implies that the estimator
$\widehat{\Xi}^{A,j}_{t^*}$ for $\Xi^{A,j}$ based on the burn-in data is
$$\widehat{\Xi}^{A,j}_{t^*} = \widehat{\nu}^j_{t^*} + \nu^j_{1} \left(\frac{1
  -\pi^j_{t^*,1}}{\pi^j_{t^*,1}} \right) + \widehat{\nu}^{j(0)}_{t^*,2} \left(
  \frac{1-\pi^{j(0)}_{t^*,2} }{\pi^{j(0)}_{t^*,2}\pi^j_{t^*,1}} \right)
  + \widehat{\nu}^{j(1)}_{t^*,2} \left(
  \frac{1-\pi^{j(1)}_{t^*,2} }{\pi^{j(1)}_{t,2}\pi^j_{t^*,1}} \right).$$
Then as for the WIPW estimator, treating $\widehat{\Xi}^{A,j}_{t^*}$ as
fixed and known, for $t \geq t^*+1$, from (\ref{eq:waipwweight}), 
 take $W^{A,j}_t =
  (\widehat{\Xi}^{A,j}_{t^*}/\widehat{\Xi}^{A,j}_{t})^{1/2}$.  Because
  $\widehat{\Xi}^{A,j}_{t^*}$ depends only on $\Data_{t^*} \subseteq
  \Data_{t-1}$ for $t \geq t^*+1$, and $\widehat{\Xi}^{A,j}_{t}$ depends only
    on $\Data_{t-1}$, the resulting weights  $W^{A,j}_t$ for $t\geq t^*+1$ depend
    on $\Data_{t-1}$.  

    Estimators for estimators for $\nu^j$, $\nu^j_{1}$, and
    $\nu^{j(s)}_{2}$, $s=0, 1$, at any time $t$ based on $\Data_{t-1}$
    can be constructed using inverse probability weighting.  With the
    current indexing scheme and with $R$ the observed response status,
    estimators can be obtained as
    $$\widehat{\nu}^j_{t} = t^{-1} \sum_{u=1}^t \frac{C^j_u
    (Y_u-\widetilde{\theta}_{t,j})^2}{ I(R_u = 0) \pi^{j(0)}_{u,2}
    \pi^j_{u,1} + I(R_u = 1) \pi^{j(1)}_{u,2} \pi^j_{u,1}},$$
$$\widehat{\nu}^j_{t,1} = t^{-1} \sum_{u=1}^t \frac{C^j_u\{ Y_u - \hatQ^j_{1,t}(\bX_{1,u}) \}^2}
{ I(R_u = 0) \pi^{j(0)}_{u,2} \pi^j_{u,1} + I(R_u = 1)
  \pi^{j(1)}_{u,2} \pi^j_{u,1}},$$ 
$$\widehat{\nu}^{j(s)}_{t,2}  = t^{-1} \sum_{u=1}^t \frac{C^j_u
  I(R_u=s) \{ Y_u - \hatQ^j_{2,t}(\Xbar_{2,u}) \}^2}{ \pi^{j(s)}_{u,2}
  \pi^j_{u,1}}, \,\,\, s=0, 1.$$
where $\widetilde{\theta}^j_{t}$ is an estimator for
$\theta^j$ using $\Data_{t-1}$ (e.g., the unweighted IPW estimator).
In the original notation, estimators can be obtained at week $t$ as
$$\widehat{\nu}^j_{t} = N^{-1}_t \sumiN \frac{\Delta_{t-1,i} C^j_i
    (Y_i-\widetilde{\theta}^j_{t})^2}{  I(R_i = 0) \pi^{j(0)}_{\tau_i,2}
    \pi^j_{\tau_i,1} + I(R_i = 1) \pi^{j(1)}_{\tau_i,2} \pi^j_{\tau_i,1}},$$
$$\widehat{\nu}^j_{t,1} =  N^{-1}_t \sumiN \frac{ \Delta_{t-1,i}C^j_i\{ Y_i - \hatQ^j_{1,t}(\bX_{1,i}) \}^2}
{  I(R_i = 0) \pi^{j(0)}_{\tau_i,2}
    \pi^j_{\tau_i,1} + I(R_i = 1) \pi^{j(1)}_{\tau_i,2} \pi^j_{\tau_i,1}},$$ 
$$\widehat{\nu}^{j(s)}_{t,2}  = N^{-1}_t \sumiN \frac{ \Delta_{t-1,i} C^j_i
  I(R_i=r) \{ Y_i - \hatQ^j_{2,t}(\Xbar_{2,i}) \}^2}{ \pi^{j(s)}_{\tau_i,2}
  \pi^j_{\tau_i,1}}, \,\,\, s=0, 1.$$

\par\noindent\textbf{\em WAIPW After Sequential Randomization}  

The considerations for obtaining the WAIPW estimator at the end of the
trial following sequential randomization are analogous to those for
WIPW.  Again take $K=2$, and consider a subject who entered the SMART
at time $\tau$ and then reached stage 2 at time $\tau+v$.  As for the
WIPW estimator, that denominators in (2) 
should be constructed using $ \pi^j_{\tau_i,1}$ and
$\pi^{j(s)}_{\tau_i+v,2}$, $s=0, 1$.  Likewise, in forming the weight
for such a subject, in (\ref{eq:hatXiAij}) use $ \pi^j_{\tau_i,1}$ and
$\pi^{j(s)}_{\tau_i+v,2}$, $s=0, 1$, and take
$$\widehat{\nu}^j_{t} = N^{-1}_t \sumiN \frac{\Delta_{t-1,i} C^j_i
    (Y_i-\widetilde{\theta}^j_{t})^2}{  I(R_i = 0) \pi^{j(0)}_{\tau_i+v,2}
    \pi^j_{\tau_i,1} + I(R_i = 1) \pi^{j(1)}_{\tau_i+v,2} \pi^j_{\tau_i,1}},$$
$$\widehat{\nu}^j_{t,1} =  N^{-1}_t \sumiN \frac{ \Delta_{t-1,i}C^j_i\{ Y_i - \hatQ^j_{1,\tau_i+v}(\bX_{1,i}) \}^2}
{  I(R_i = 0) \pi^{j(0)}_{\tau_i+v,2}
    \pi^j_{\tau_i,1} + I(R_i = 1) \pi^{j(1)}_{\tau_i+v,2} \pi^j_{\tau_i,1}},$$ 
$$\widehat{\nu}^{j(s)}_{t,2}  = N^{-1}_t \sumiN \frac{ \Delta_{t-1,i} C^j_i
  I(R_i=r) \{ Y_i - \hatQ^j_{2,\tau_i+v}(\Xbar_{2,i}) \}^2}{ \pi^{j(s)}_{\tau_i+v,2}
  \pi^j_{i,1}}, \,\,\, s=0, 1.$$

\par\noindent\textbf{\em Verifying Assumptions for the WIPW and WAIPW Estimators}

\newcommand{\calUThetanb}{\mathcal{M}_T(\theta)}
\newcommand{\calUThetaStarnb}{\mathcal{M}_T(\theta^*)}
\newcommand{\calUThetaBarTnb}{\mathcal{M}_T(\overline{\theta}_T)}
\newcommand{\xiStarNegHalfnb}{\xi_T^{-1/2}(\theta^*)}
\newcommand{\xiStarnb}{\bxi_T(\theta^*)}
\newcommand{\xiStarHalfnb}{\xi_T^{1/2}(\theta^*)}

First, denote the following estimating functions corresponding to the
WIPW and WAIPW estimators, respectively:

\begin{align*}
M_t^{WIPW,j}& (\Xbar_{K,t},\Abar_{K,t},Y_t; \theta^j) = \frac{ W^j_t
  C^j_t}{ \{ \prod_{k=2}^K \pi_{t,k}^j(\Xbar_{k,t})\}
  \pi_{t,1}^j(\bX_{1,t}) } (Y_t - \theta^j), \\
  M_t^{WAIPW,j}& (\Xbar_{2,t},\Abar_{2,t},Y_t; \theta^j) 
  = W^{A,j}_t \left[ \frac{C^j_t
           Y_t}{\pi^j_{t,2}(\Xbar_{2,t})\pi^j_{t,1}(\bX_{1,t})} -\left\{
           \frac{\overline{C}^j_{1,t}}{\pi^j_{t,1}(\bX_{1,t})} -1 \right\}    \hatQ^j_{1,t}(\bX_{1,t}) \right.\\
   &-\left. \left\{ \frac{C^j}{\pi^j_{t,2}(\Xbar_{2,t})\pi^j_{t,1}(\bX_{1,t}1)} -\frac{\overline{C}^j_{1,t}}{\pi^j_{t,1}(\bX_{1,t})} \right\}
     \hatQ^j_{2,t}(\Xbar_{2,t}) - \theta^j\right].
  \end{align*}
\noindent
Note that these are estimating $\theta^j \in \R$, so $p=1$. To reduce
the notational burden, we use the short-hand notation
$M_t^{WIPW,j}(\theta^j)$ and $M_t^{WAIPW,j}(\theta^j)$. Additionally,
denote
$\mathcal{M}_t^{WIPW,j}(\theta^j) = \sum_{t=1}^T
M_t^{WIPW,j}(\theta^j)$ and
$\mathcal{M}_t^{WAIPW,j}(\theta^j) = \sum_{t=1}^T
M_t^{WAIPW,j}(\theta^j)$. We introduce the notation $\theta^{j*}$ to
denote the true parameter, whereas $\theta^j$ for a generic value.

We now consider assumptions (A0)-(A8) in Appendix B and demonstrate that each assumption holds for the foregoing estimating functions. 

For the first assumption (A0), we verify that
$\{\mathcal{M}^{WIPW,j}_t(\theta^{j*})\}_{t\geq1}$ and
$\{\mathcal{M}^{WAIPW,j}_t(\theta^{j*})\}_{t\geq1}$ are a Martingale
difference sequences with respect to $\{ \mathcal{F}_{t}\}_{t\geq1}$
where $\mathcal{F}_{t}=\Data_t$. As shown in previous sections,
$E\left\{ M_t^{WIPW}(\theta^{j*}) \mid \Data_{t-1} \right\} = 0$ and
$E\left\{ M_t^{WAIPW}(\theta^{j*}) \mid \Data_{t-1} \right\} = 0$,
which implies that both are Martingale difference sequences.

Assumption (A1) is that $\mathcal{M}(\theta)$ is continuously differentiable with respect to $\theta$ is invertible. We trivially assume that $C_t^j \neq 0$ for all $t$. Consider the derivatives:
\begin{align*}
\nabla_{\theta}\mathcal{M}_{T}^{WIPW,j}(\theta^j) &= -\sum_{t=1}^T \frac{ W^j_t
  C^j_t}{ \{ \prod_{k=2}^K \pi_{t,k}^j(\Xbar_{k,t})\}
  \pi_{t,1}^j(\bX_{1,t}) }, \\
\nabla_{\theta} \mathcal{M}_{T}^{WIPW,j}(\theta^j) &= -\sum_{t=1}^T W_t^{A,j}.
\end{align*}

The positivity assumption implies
$\{ \prod_{k=2}^K \pi_{t,k}^j(\Xbar_{k,t})\} \pi_{t,1}^j(\bX_{1,t})>0$
for all $t$. By construction, $W_t^j\neq0$ and $W_t^{A,j}\neq 0$ for
all $t$ (e.g., burn-in period), which verifies (A1).

Assumption (A2) states that
$\lambda_{min} \left[ \sum_{t=1}^T \mbox{var}\{ M_t(\theta^*) \mid
  \mathcal{F}_{t-1} \} \right] \longrightarrow \infty$. Because we are
estimating a scalar $\theta$, this reduces to
$\sum_{t=1}^T \mbox{var}\{ M_t(\theta^*) \mid \mathcal{F}_{t-1} \}
\longrightarrow \infty$. Because of how we choose the weights,
$E \left[ \{ M^{WIPW,j}_t(\theta^{j*})\}^2 \mid \Data_{t-1} \right] =
\sigma > 0$ and likewise for
$E\left[ \{ M^{WAIPW,j}_t(\theta^{j*}) \}^2 \mid \Data_{t-1}
\right]$. Hence, (A2) holds for both estimators. As mentioned in
Section~\ref{ss:thompson}, a clipping constant may be necessary to
ensure the positivity assumption. The positivity assumption implies
$E(C_t^j \mid \mathcal{W}^*,\Data_{t-1}) > 0$ with probability one for
all $t$, which is necessary for (A2) to hold.

Assumption (A3) states:
\begin{equation*}
  \sup_{\theta\in\bTheta} \bigg|\bigg|
  \left\lbrace 
    \xiStarNegHalfnb \nabla_{\theta}\mathcal{M}_T(\theta) 
    \xiStarNegHalfnb 
  \right\rbrace^{-1}
  \bigg|\bigg| = O_P(1).   
\end{equation*}
\noindent

Let
$I^{WIPW,j}_t(\epsilon) = I( W^j_t C^j_t / \{ \prod_{k=2}^K
\pi_{t,k}^j(\Xbar_{k,t}) \}\pi_{t,1}^j(\bX_{1,t}) > \epsilon)$. With
the positivity assumption, the trivial assumption that $C_j^t \neq 0$
for all $t$, and because $W_t^j>0$ as long as
$Y_t\neq \widetilde{\theta}^j$, which occurs infinitely often in
non-degenerate distributions on $Y$, we can use the filtration
Borel-Cantelli Lemma to determine there exists a $\theta>0, \tau>0$
such that $P\{ I^{WIPW,j}(\epsilon) \mid \Data_{t-1} \} > \tau$. This
implies
$-\sum_{t=1}^T W^j_t C^j_t / \{ \prod_{k=2}^K \pi_{t,k}^j(\Xbar_{k,t})
\} \pi_{t,1}^j(\bX_{1,t}) = O_P(T)$. Define
$I^{WAIPW,j}(\epsilon) = I(W^{A,j}_t > \epsilon)$ and take similar
steps to determine $-\sum_{t=1}^T W^{A,j}_t = O_P(T)$.

it holds that $W^j_t C^j_t / \{ \prod_{k=2}^K \pi_{t,k}^j(\Xbar_{k,t}) \}\pi_{t,1}^j(\bX_{1,t}) > 0$ and $W_t^{A,j} > 0$ occur with non-zero probability. 

For the WIPW estimator,
\begin{align*}
\sup_{\theta\in\bTheta} \bigg|\bigg|
    \left\lbrace 
    \xiStarNegHalfnb \nabla_{\theta}\mathcal{M}_t(\theta) 
    \xiStarNegHalfnb 
    \right\rbrace^{-1}
\bigg|\bigg| &= \bigg|\bigg| \frac{\sum_{t=1}^T E\left[ \{ M^{WIPW,j}_t(\theta^{j*}) \}^2 \mid \Data_{t-1} \right]}{ -\sum_{t=1}^T W^j_t
  C^j_t / \{ \prod_{k=2}^K \pi_{t,k}^j(\Xbar_{k,t}) \}
  \pi_{t,1}^j(\bX_{1,t}) }  \bigg|\bigg| \\
&= \bigg|\bigg| \frac{T\sigma}{O_P(T) }  \bigg|\bigg| \\
&= \bigg|\bigg| \frac{O_P(T)}{O_P(T) }  \bigg|\bigg| \\
&= O_P(1),
\end{align*}
\noindent 
and similarly, for the WAIPW estimator,
\begin{align*}
\sup_{\theta\in\bTheta} \bigg|\bigg|
    \left\lbrace 
    \xiStarNegHalfnb \nabla_{\theta}\mathcal{M}_t(\theta) 
    \xiStarNegHalfnb 
    \right\rbrace^{-1}
\bigg|\bigg| &= \bigg|\bigg| \frac{\sum_{t=1}^T E\left[ \left\{ M^{WAIPW,j}_t(\theta^{j*}) \right\}^2 \mid \Data_{t-1} \right]}{ -\sum_{i=1} W^{A,j}_t }  \bigg|\bigg| \\
&= \bigg|\bigg| \frac{T\sigma}{O_P(T)} \bigg|\bigg| \\
&= \bigg|\bigg| \frac{O_P(T)}{O_P(T) }  \bigg|\bigg| \\
&= O_P(1).
\end{align*}

Assumption (A4) states: $\xiStarNegHalfnb \calUThetaStarnb = O_P(1)$.      

For the WIPW estimator
\begin{align*}
\xiStarNegHalfnb \calUThetaStarnb &= \frac{\sum_{t=1}^T M^{WIPW,j}_t(\theta^{j*}) }{ \left( \sum_{t=1}^T E\left[ \left\{ M^{WIPW,j}_t(\theta^{j*}) \right\}^2 \mid \Data_{t-1} \right] \right)^{1/2} } \\
&= \frac{T^{-1} \sum_{t=1}^T M^{WIPW,j}_t(\theta^{j*}) }{T^{-1} \left( \sum_{t=1}^T E\left[ \left\{ M^{WIPW,j}_t(\theta^{j*}) \right\}^2 \mid \Data_{t-1} \right] \right)^{1/2} } \\
&= \frac{T^{-1} \sum_{t=1}^T M^{WIPW,j}_t(\theta^{j*}) }{\sigma^{1/2} } \\
&= O_P(1),
\end{align*}
\noindent
because $T^{-1} \sum_{t=1}^T M^{WIPW,j}_t(\theta^{j*}) \inp 0$, a
result of the strong law of large numbers for Martingale difference
sequences \citep{slln_mds}.

Similarly, for the WAIPW estimator,
\begin{align*}
\xiStarNegHalfnb \calUThetaStarnb &= \frac{\sum_{t=1}^T M^{WAIPW,j}_t(\theta^{j*}) }{ \left( \sum_{t=1}^T E\left[ \{ M^{WIPW,j}_t(\theta^j) \}^2 \mid \Data_{t-1} \right] \right)^{1/2} } \\
&= \frac{T^{-1}\sum_{t=1}^T M^{WAIPW,j}_t(\theta^{j*}) }{ T^{-1} \left( \sum_{t=1}^T E\left[ \{ M^{WIPW,j}_t(\theta^j) \}^2 \mid \Data_{t-1} \right] \right)^{1/2} } \\
&= \frac{T^{-1}\sum_{t=1}^T M^{WAIPW,j}_t(\theta^{j*}) }{ \sigma^{1/2} } \\
&= O_P(1).
\end{align*}

Assumption (A5) states
$\xiStarNegHalfnb \mathcal{M}_T(\theta^*) \inD \mathcal{N}(0,
1)$. Because $M_t^{WIPW,j}(\theta^{j*})$ and
$M_t^{WAIPW,j}(\theta^{j*})$ are Martingale difference sequences (A0),
and
$\Sigma_{t=1}^T E\left[ \{ M_t^{WIPW,j}(\theta^{j*}) \}^2 \mid
  \Data_{t-1} \right] \longrightarrow \infty$ with probability one,
and likewise for WAIPW (A2) the Martingale central limit theorem
\citep{Hall_Heyde_1980} guarantees:
\begin{align*}
& \frac{\sum_{t=1}^T M_t^{WIPW,j}(\theta^{j*})}{ \left( \sum_{t=1}^T E\left[ \{ M^{WIPW,j}_t(\theta^j) \}^2 \mid \Data_{t-1} \right] \right)^{1/2}}  \inD \mathcal{N}(0, 1), \\
& \frac{\sum_{t=1}^T M_t^{WAIPW,j}(\theta^{j*})}{\left( \sum_{t=1}^T E\left[ \{ M^{WAIPW,j}_t(\theta^j) \}^2 \mid \Data_{t-1} \right] \right)^{1/2}} \inD \mathcal{N}(0, 1). 
\end{align*}

Assumption (A6) states for any consistent estimator: $\widetilde{\theta}_T \inp \theta^*$  
 \begin{equation*}
    \bigg|\bigg| 
    \xiStarNegHalfnb \left\lbrace
        \nabla_{\theta}\calU(\widetilde{\theta}_T) - 
        \nabla_{\theta} \calUThetaStarnb 
    \right\rbrace
    \xiStarNegHalfnb
    \bigg|\bigg| \inp 0.
 \end{equation*}

For both WIPW and WAIPW, $\nabla_{\theta}\calU(\widetilde{\theta}^j_T) = \nabla_{\theta}\calU(\theta^{j*})$, so this holds.

Assumption (A7) states that there exists fixed and positive value $\sigma'$ such that 
 \begin{equation*}
     -\xiStarNegHalfnb \nabla_{\theta} \calUThetaStarnb \xiStarNegHalfnb 
     \inp \sigma',
  \end{equation*}

\noindent
Due to the Martingale convergence theorem, there exists some $\tau < \infty$ such that the following hold: $T^{-1}\sum_{t=1}^T E\left[ \{ M^{WIPW,j}_t(\theta^{j*}) \}^2 \mid \Data_{t-1} \right] \inp \sigma$ and $T^{-1} \sum_{t=1}^T W^{A,j}_t \inp \tau$.

For the WIPW estimator,
\begin{align*}
-\xiStarNegHalfnb \nabla_{\theta} \calUThetaStarnb \xiStarNegHalfnb &= \frac{T^{-1} \sum_{t=1}^T W^j_t
  C^j_t / \{ \prod_{k=2}^K \pi_{t,k}^j(\Xbar_{k,t}) \}
  \pi_{t,1}^j(\bX_{1,t})}{T^{-1}\sum_{t=1}^T E\left[ \{ M^{WIPW,j}_t(\theta^{j*}) \}^2 \mid \Data_{t-1} \right]  } \\
&= \frac{T^{-1} \sum_{t=1}^T W^j_t
  C^j_t / \{ \prod_{k=2}^K \pi_{t,k}^j(\Xbar_{k,t}) \}
  \pi_{t,1}^j(\bX_{1,t})}{\sigma } \\
&\inp \frac{\tau}{\sigma } \\
&>0.
\end{align*}
\noindent
Similarly, for the WAIPW estimator,
\begin{align*}
-\xiStarNegHalfnb \nabla_{\btheta} \calUThetaStarnb \xiStarNegHalfnb &= \frac{T^{-1} \sum_{t=1}^T W^{A,j}_t
  }{T^{-1}\sum_{t=1}^T E\left[ \{ M^{WAIPW,j}_t(\theta^j) \}^2 \mid \Data_{t-1} \right]  } \\
&= \frac{T^{-1} \sum_{t=1}^T W^{A,j}_t }{\sigma } \\
&\inp \frac{\tau' }{\sigma } \\
&>0
\end{align*}

Assumption (A8) states $\widehat{\theta}_T$ satisfies $\xiStarNegHalfnb \calU(\widehat{\theta}_T)\inp 0$ as $T\rightarrow \infty$.

Both $E\left\{ M_t^{WIPW,j}(\widehat{\theta}_t^{WIPW,j}) \right\}=0$ and $E\left\{ M_t^{WAIPW,j}(\widehat{\theta}_t^{WAIPW,j}) \right\}=0$.

For WIPW, due to the the strong law of large numbers of Martingale differences:
\begin{align*}
\xiStarNegHalfnb \calU(\widehat{\theta}_T) &= \frac{T^{-1} \mathcal{M}_t^{WIPW,j}(\widehat{\theta}_t^j)}{T^{-1}\sum_{t=1}^T E\left[ \left\{M^{WIPW,j}_t(\widehat{\theta}^{WIPW,j}_t)\right\}^2 \mid \Data_{t-1} \right]} \\
&\inp 0/\sigma
\end{align*}
\noindent
and likewise for WAIPW: 
\begin{align*}
\xiStarNegHalfnb \calU(\widehat{\theta}_T) &= \frac{T^{-1} \mathcal{M}_t^{WAIPW,j}(\widehat{\theta}_t^j)}{T^{-1}\sum_{t=1}^T E\left[ \left\{M^{WAIPW,j}_t(\widehat{\theta}^{WAIPW,j}_t)\right\}^2 \mid \Data_{t-1} \right]} \\
&\inp 0/\sigma
\end{align*}


\end{document}